\documentclass[10pt,article,oneside]{memoir}

\usepackage[english]{babel}%langue
\usepackage[utf8]{inputenc}
\usepackage[T1]{fontenc}%encodage
\usepackage{calligra}%Pour certains signes
\usepackage{amsfonts,amsmath,amssymb}%maths
\usepackage{amsthm}
\usepackage{mathabx}%maths
\usepackage[usenames]{color} % Pour dä?ºinir les couleurs
\usepackage{xspace}%Pour les commandes de nom de joueurs
\usepackage[svgnames]{xcolor}
\usepackage{subfig,caption}
\usepackage{todonotes}
\usepackage{paralist}
\usepackage{authblk}
\usepackage{microtype}
\usepackage{tikz}
\usetikzlibrary{shapes}

\usepackage[margin=2.5cm]{geometry}
\newif\ifdraft\draftfalse
\newif\ifshort\shortfalse
\counterwithout{section}{chapter}

\ifdraft

\newcommand\os[1]{\todo[inline,size=\scriptsize,backgroundcolor=PaleTurquoise]{#1 - \textbf{Olivier}}}
\newcommand\ah[1]{\todo[inline,size=\scriptsize,backgroundcolor=Yellow]{#1 - \textbf{Axel}}}
\newcommand\ac[1]{\todo[inline,size=\scriptsize,backgroundcolor=SpringGreen]{#1 - \textbf{Arnaud}}}
\newcommand{\acchanged}[1]{{\color{Green}{#1}}}
\newcommand\oschanged[1]{{\color{RoyalBlue}{#1}}}
\newcommand\vlong[1]{\todo[inline,size=\scriptsize,backgroundcolor=red, caption={2do}]{
\begin{minipage}{\textwidth-4pt}#1  - \textbf{Garder dans la version longue}\end{minipage}}}
\else

\newcommand\os[1]{}
\newcommand\ac[1]{}
\newcommand\ah[1]{}
\newcommand\review[1]{}
\newcommand\oschanged[1]{#1}
\newcommand\acchanged[1]{#1}

\newcommand\vlong[1]{}
\fi

\newcommand{\ie}{\emph{i.e.}\xspace}
\newcommand{\aka}{a.k.a.\xspace}
\renewcommand{\epsilon}{\varepsilon}
\renewcommand{\phi}{\varphi}

\newcommand{\N}{\mathbb{N}}%Naturels
\newcommand{\A}{\mathcal{A}}%Automate
\newcommand{\B}{\mathcal{B}}%Ensemble des branches
\newcommand{\T}{ t }%Arbre
\newcommand{\Branches}{ \{0,1\}^{\omega} }%Noeuds
\newcommand{\Acc}{\mathrm{Acc}}

\newtheorem{theorem}{Theorem}

\newtheorem{proposition}{Proposition}
\newtheorem{lemma}{Lemma}

\newtheorem{remark}{Remark}

\makeatletter
\def\timenow{\@tempcnta\time
  \@tempcntb\@tempcnta
  \divide\@tempcntb60
  \ifnum10>\@tempcntb0\fi\number\@tempcntb
  \multiply\@tempcntb60
  \advance\@tempcnta-\@tempcntb
  :\ifnum10>\@tempcnta0\fi\number\@tempcnta}
\makeatother
\newcommand\textbfit[1]{\textbf{\em #1}}
\newcommand{\defin}[1]{\textbfit{\boldmath #1}}

\newcommand{\qini}{q_{\mathrm{ini}}}

\newcommand{\col}{\mathrm{Col}}
\newcommand{\colors}{C}
\newcommand{\nat}{\mathbb{N}}

\newcommand{\prefix}{\sqsubseteq}
\newcommand{\prefixstrict}{\sqsubset}

\newcommand{\Eloise}{\'Elo\"ise\xspace}
\newcommand{\Abelard}{Ab\'elard\xspace}

\newcommand{\Ei}{\mathbf{E}}
\newcommand{\Ai}{\mathbf{A}}

\newcommand{\arena}{\mathcal{G}}
\newcommand{\game}{\mathbb{G}}
\newcommand{\play}{\lambda}
\newcommand{\Plays}[1]{\mathrm{Plays}}
\newcommand{\Traces}[1]{\mathrm{Traces}}

\newcommand{\branch}{\pi} 

\newcommand{\run}{\rho}

\newcommand{\Evei}{\Ei}
\newcommand{\Adami}{\Ai}

\newcommand{\LRejAtMostCount}[1]{L^{\mathrm{Rej}}_{\leq \mathrm{Count}}(#1)}
\newcommand{\LAccInf}[1]{L^{\mathrm{Acc}}_{\infty}(#1)}
\newcommand{\LAccUnc}[1]{L^{\mathrm{Acc}}_{Uncount}(#1)}
\newcommand{\LRejAtMostFin}[1]{L^{\mathrm{Rej}}_{\mathrm{Fin}}(#1)}

\newcommand{\gameRejFin}[1]{\game^{\mathrm{Rej} \mathrm{Fin}}_{#1}}
\newcommand{\gameRejCount}[1]{\game^{\mathrm{Rej}\leq \mathrm{Count}}_{#1}}
\newcommand{\gameRejCountB}[1]{\game'^{\mathrm{Rej}\leq \mathrm{Count}}_{#1}}
\newcommand{\gameAccUnc}[1]{\game^{\mathrm{Acc}\, \mathrm{Uncount}}_{#1}}
\newcommand{\gameAccInf}[1]{\game^{\mathrm{Acc}\, \infty}_{#1}}
\newcommand{\gameAccUnck}[1]{\widetilde{\game}^{\mathrm{Acc}\, \mathrm{Uncount}}_{#1}}

\newcommand{\VE}{V_{\Evei}}
\newcommand{\VA}{V_{\Adami}}
\newcommand{\WC}{\Omega}

\newcommand{\eg}{\emph{e.g.}\xspace}
\newcommand{\resp}{resp.\xspace} % Pas du latin, donc en gÃ©nÃ©ral pas en italique (aprÃ¨s Ã§a me choque pas que Ã§a soit en italique) AH
\newcommand{\set}[1]{\{#1\}}

\newcommand{\Brchs}{\mathcal{\B}}

\newcommand{\LLarge}[1]{L^{\mathrm{Acc}}_{\mathrm{Large}}(#1)}
\newcommand{\gameAccLarge}[1]{\game^{\mathrm{Acc}\, Large}_{#1}}

\title{Counting Branches in Trees Using Games}
\author[1]{Arnaud Carayol}
\author[2]{Axel Haddad}
\author[3]{Olivier Serre}

\affil[1]{LIGM (CNRS \& Université Paris Est)}
\affil[2]{Université de Mons}
\affil[3]{LIAFA (CNRS \& Université Paris Diderot -- Paris 7)}

\begin{document}

\maketitle

\begin{abstract}
{
We study finite automata running over infinite binary trees. \acchanged{A run of such an automaton is usually said to be accepting if all its branches are accepting. In this article, we relax the notion of accepting run by allowing a certain quantity of rejecting branches.}  More precisely we study the following criteria for a run to be accepting: 
\begin{enumerate}[(i)]
\item it contains at most finitely (\resp countably) many rejecting branches;
\item it contains infinitely (\resp uncountably) many accepting branches;
\item the set of accepting branches is topologically “big”.
\end{enumerate}
In all situations we provide a simple acceptance game that later permits to prove that the languages accepted by automata with cardinality constraints are always $\omega$-regular.
In the case (ii) where one counts accepting branches it leads to new proofs (without appealing to logic) of an old result of Beauquier and Niwi\'nski.}
\end{abstract}

%%%%%%%%%%%%%%%%%% INTRO %%%%%%%%%%%%%%%%%%%%

\section{Introduction}

There are several natural ways of describing sets of infinite trees. One is \emph{logic} where, with any formula, is associated the set of all trees for which the formula holds. Another option is using \emph{finite automata}. 
Finite automata on infinite trees (that extends both automata on infinite words and on finite trees) were originally introduced by Rabin in \cite{Rabin69} to prove the decidability of the monadic second order logic (MSOL) over the full binary tree. Indeed, Rabin proved that for any MSOL formula, one can construct a tree automaton such that it accepts a non empty language if and only if the original formula holds at the root of the full binary tree. 
%\ac{On peut enlever la phrase suivante.}
{These automata were also successfully used by Rabin in \cite{Rabin72} to solve Church's synthesis problem \cite{Church62}, that  asks for constructing a circuit based on a formal specification (typically expressed in MSOL) describing the desired input/output behaviour. }
His approach was to represent the set of all possible behaviours of a circuit by an infinite tree (directions  code the inputs while node labels along a branch code the outputs) and to reduce the synthesis problem to emptiness of a tree automaton accepting all those trees coding circuits satisfying the specification.
Since then, automata on infinite trees and their variants have been intensively studied and found many applications, in particular in logic. Connections between automata on infinite trees and logic are discussed \eg in the excellent surveys \cite{Thomas97,VW07}.
%\os{Axel avait proposé d'enlever excellent; je trouve que c'est un peu injuste en fait et on gagne pas de place… Si on veut gagner on peut virer la deuxième ref ce qui permet d'économiser une entrée dans la biblio}

Roughly speaking a finite automaton on infinite trees is a finite memory machine that takes as input an infinite node-labelled binary tree and processes it in a top-down fashion as follows. It starts at the root of the tree in its initial state, and picks (possibly nondeterministically) two successor states, one per son, according to the current control state, the letter at the current node and the transition relation. Then the computation proceeds in parallel from both sons, and so on. Hence, a run of the automaton on an input tree is a labelling of this tree by control states of the automaton, that should satisfy the local constrains imposed by the transition relation. A branch in a run is accepting if the $\omega$-word obtained by reading the states along the branch satisfies some acceptance condition (typically an $\omega$-regular condition such as a B\"uchi or a parity condition). Finally, a tree is accepted by the automaton if {there exists} a run over this tree in which \emph{every} branch is accepting. An $\omega$-regular tree language is a tree language accepted by some tree automaton equipped with a parity condition.

A fundamental result of Rabin is that $\omega$-regular tree languages form a Boolean algebra \cite{Rabin69}.   \acchanged{The main technical difficulty in establishing this result is to show the closure under complementation.}
%The hard part in his proof is  complementation, and 
Since the publication of this result in 1969, it has been a challenging problem to simplify \acchanged{this} proof. A much simpler one was obtained by Gurevich and Harrington in \cite{GurevichH82} making use of two-player perfect information games for checking membership of a tree in the language accepted by the automaton (note that the idea of using games to prove this result was already proposed by Büchi in \cite{Buchi77}):
 \Eloise (\aka \emph{Automaton}) builds a run on the input tree while \Abelard (\aka \emph{Pathfinder}) tries to exhibit a rejecting branch in the run. 
Another fruitful connection between automata and games is for  emptiness checking. 
%as the emptiness problem can be reduced to solving a two-player \emph{finite} game. 
In a nutshell the emptiness problem for an automaton on infinite trees can be modelled as a game where \Eloise builds an input tree together with a run while \Abelard tries to exhibit a rejecting branch in the run. 
Hence, the emptiness problem for tree automata can be reduced to solving a two-player parity game played on a \emph{finite} graph. 
Beyond these results, the tight connection between automata and games is one of the main tools in automata theory \cite{Thomas97,GWT02,Loeding11}. 

There are several levers on which one can act to define alternative families of tree automata / classes of tree languages. A first lever is \emph{local} with respect to the run: it is the condition required for a branch to be accepting, the reasonable options here being all classical $\omega$-regular conditions (reachability, Büchi, parity\ldots). A second one has to do with the set of runs. The usual definition is existential: a tree is accepted if there exists an accepting run on that tree. Other popular approaches are universality, alternation or probabilistic. A third lever is \emph{global} with respect to the run: it is the condition required for a run to be accepting. The usual definition is that \emph{all} branches must be accepting for the run to be accepting but one could  relax this condition by specifying \emph{how many} branches should be accepting/rejecting. One can do this either by \emph{counting} the number of accepting branches (\eg infinitely many, uncountably many)  or by counting the number of rejecting branches (\eg finitely many, at most countably many): this leads to the notion of automata with cardinality constraints \cite{BeauquierNN91,BN95}. {As these properties can be expressed in MSOL \cite{BKR10}, the classes of languages accepted under these various restrictions are always $\omega$-regular. However, this logical approach does not give a tractable transformation to standard parity or Büchi automata.}   {Another option is to use a notion of topological “bigness” and to require for a run to be accepting that the set of accepting branches is big}~\cite{VVK05,ACV10}.  Yet another option considered in \cite{CHS11,FPS13} is to \emph{measure} (in the usual sense of measure theory) the set of accepting branches and to put a constraint on this measure (\eg positive, equal to one).
%at least one half, equal to one

The idea of allowing a certain amount of rejecting branches in a run was first considered by Beauquier, Nivat and Niwi{\'n}ski in \cite{BeauquierNN91,BN95}, where it was required that the number of accepting branches in a run belongs to a specified set of cardinals $\Gamma$. In particular, they proved that if $\Gamma$ consists of all cardinals greater than some $\gamma$, then one obtains an $\omega$-regular tree language. Their approach was based on logic (actually they proved that a tree language defined by such an automaton can be defined by a $\Sigma_1^1$ formula hence, can also be defined by a \emph{Büchi} tree automaton) while the one we develop here is based on designing acceptance games. 
There is also work on the logical side with decidable results but that do not lead to efficient algorithms~\cite{BKR10}.%Note that we also consider the case where one requires that there are at most finitely many (\resp countably many) rejecting branches in a run to be accepting. \os{Parler de \cite{BKR10}}
%For each of the various classes, we provide an acceptance game that can later be modified to obtain an \emph{equivalent} game that comes from a tree automaton with the classical semantics. 

Our main contributions are to introduce (automata with cardinality constraints on the number of rejecting branches; automata with topological bigness constraints) or revisit (automata with cardinality constraints on the number of accepting branches) variants of tree automata where acceptance for a run allows a somehow negligible set of rejecting branches. For each model, we provide a game counterpart by mean of an equivalent acceptance game and this permits to retrieve the classical (and fruitful) connection between automata and game. It also permit to argue that languages defined by those classes are always $\omega$-regular. Moreover, in the case where one counts accepting branches we show that the  languages that we obtain are always accepted by a \emph{Büchi} automaton, which contrasts with the case where one counts rejecting branches where we exhibit a counter-example for that property. 

The paper is organised as follows. Section~\ref{sec:def} recalls classical concepts while Section~\ref{sec:autodef} introduces the main notions studied in the paper, namely automata with cardinality constraints and automata with topological bigness constraints. Then, Section~\ref{sec:results-rejecting} studies those languages obtained by automata with cardinality constraints on the number of rejecting branches while Section~\ref{sec:results-accepting} is devoted to those languages obtained by automata with cardinality constraints on the number of accepting branches. Finally, Section~\ref{sec:results-topo} considers automata with topological bigness constraints.

%%%%%%%%%%%% DEFINITIONS %%%%%%%%%%%%%%%%%%%%%%%

\section{Preliminaries}\label{sec:def}

\subsection{Words and Trees}

An \defin{alphabet} $A$ is a (possibly infinite) set of letters. In the sequel $A^*$ denotes the set of finite words over $A$, and $A^\omega$ the set of infinite words (or $\omega$-words) over $A$. The \defin{empty word} is written $\epsilon$. The length of a word $u$ is denoted by $|u|$. For any $k\geq 0$, we let $A^k=\{u\mid |u|=k\}$, $A^{\leq k}=\{u\mid |u|\leq k\}$ and $A^{\geq k}=\{u\mid |u|\geq k\}$. We let $A^+=A^*\setminus\{\epsilon\}$.

%\os{Vérifier que l'on se sert des notions de préfixe et de préfixe strict}
 Let $u$ be a finite word and $v$ be a (possibly infinite) word. Then $u\cdot v$ (or simply $uv$) denotes the \defin{concatenation} of $u$ and $v$; the word $u$ is a \defin{prefix} of $v$, denoted $u\prefix v$, if there exists a word $w$ such that $v=u\cdot w$. We denote by $u \prefixstrict v$ the fact that $u$ is a strict prefix of $v$ (\ie\ $u \prefix v$ and $u \not=v$). For some word $u$ and some integer $k\geq 0$, we denote by $u^k$ the word obtained by concatenating $k$ copies of $u$ (with the convention that $u^0=\epsilon$).

In this paper we consider full binary node-labelled trees. Let $A$ be an alphabet, then an \defin{$A$-labelled tree} $\T$ is a (total) function from  $\{0,1\}^{*}$ to $A$. In this context, an element $u\in\{0,1\}^*$ is called a \defin{node}, and the node $u\cdot 0$ (\emph{resp.} $u\cdot 1$) is the \defin{left son} (\emph{resp.} \defin{right son}) of $u$. The node $\varepsilon$ is called the \defin{root}.
% We shall refer to $|u|$ as the \defin{depth} of $u$.  
The letter $t(u)$ is called the \defin{label} of $u$ in $\T$. 

A \defin{branch} is an infinite word $\branch\in\{0,1\}^\omega$ and a node $u$ belongs to a branch $\branch$ if $u$ is a prefix of $\branch$. For an $A$-labelled  tree $\T$ and a branch $\branch=\branch_0 \branch_1\cdots$ we define the \defin{label} of $\branch$ as the $\omega$-word $\T(\branch)=\T(\epsilon)\T(\branch_0)\T(\branch_0\branch_1)\T(\branch_0\branch_1\branch_2)\cdots$.

%Given a tree $\T$ and a node $u$, the \defin{subtree of $t$ rooted at $u$} denoted $\subtree{\T}{u}$ is the tree defined by $\subtree{\T}{u}(v)=\T(u\cdot v)$. A tree $\T$ is said to be \defin{regular} if it contains only finitely many different subtrees, \emph{i.e.} the set $\{\subtree{\T}{u}\mid u\in\{0,1\}^*\}$ is finite.

\subsection{Two-Player Perfect Information Turn-Based Games on Graphs}\label{sec:prelim:games}

A \defin{graph} is a pair $G=(V,E)$ where $V$ is a (possibly infinite) set of
\defin{vertices} and $E\subseteq V\times V$ is a set of \defin{edges}. For a vertex $v$ we let $E(v)=\{v'\mid (v,v')\in E\}$ and in the rest of the paper (hence, this is implicit from now on), we only consider graphs that have no dead-end, \ie such that $E(v)\neq\emptyset$ for all $v$.

%A \defin{dead-end} is a vertex $v$ such that there is no vertex $v'$ with $(v,v')\in E$; in the rest of the paper, we only consider graphs that have no dead-end, hence this is implicit from now on. For a vertex $v$ we let $E(v)=\{v'\mid (v,v')\in E\}$.

An \defin{arena} is a triple $\arena=(G,\VE,\VA)$ where $G=(V,E)$ is a
graph and $V=\VE\uplus\VA$ is a partition of the vertices among two
players, \Eloise and \Abelard.

\Eloise and \Abelard play in $\arena$ by moving a pebble along edges. A
\defin{play} from an initial vertex $v_0$ proceeds as
follows: the player owning $v_0$ (\ie \Eloise if $v_0\in \VE$,
\Abelard otherwise) moves the pebble to a vertex $v_1\in E(v_0)$. Then the player owning $v_1$ chooses a successor
$v_2\in E(v_1)$ and so on. As we assumed that there is no dead-end, a play is an
infinite word $v_0v_1v_2\cdots \in V^\omega$ such that for all $0\leq i$ one has $v_{i+1}\in E(v_i)$. A \defin{partial play} is a prefix of a play, \ie, it is a finite word $v_0v_1\cdots v_\ell \in V^*$  such that for all $0\leq i<\ell$ one has $v_{i+1}\in E(v_i)$.

A \defin{strategy} for \Eloise is a function $\phi:V^*V_\Ei\rightarrow V$ assigning, to every partial play
ending in some vertex $v\in \VE$, a vertex $v'\in E(v)$. Strategies of \Abelard are defined likewise, and usually denoted $\psi$.
In a given play $\lambda=v_0v_1\cdots$ we say that \Eloise (\resp \Abelard) \defin{respects a strategy} $\phi$ (\resp $\psi$) if whenever $v_i\in V_\Ei$ (\resp $v_i\in V_\Ai$) one has $v_{i+1} = \phi(v_0\cdots v_i)$ (\resp $v_{i+1} = \psi(v_0\cdots v_i)$).

A \defin{winning condition} is a subset $\WC\subseteq V^\omega$ and a (two-player perfect information) \defin{game} is a pair $\game=(\arena,\WC)$ consisting of an arena and a winning condition.% A game is finite if it is played on a finite arena.

A play $\lambda$ is \defin{won} by \Eloise if and only if $\lambda\in \WC$; otherwise $\lambda$ is won by \Abelard. A strategy $\phi$ is \defin{winning} for \Eloise in $\game$ from a vertex $v_0$ if any play starting from $v_0$ where \Eloise respects $\phi$ is won by her. Finally, a vertex $v_0$ is \defin{winning} for \Eloise in $\game$ if she has a winning strategy $\phi$ from $v_0$. Winning strategies and winning vertices for \Abelard are defined likewise.

%%%

We now define three classical winning conditions.
\begin{itemize}
%\begin{inparaenum}[(\itshape i\upshape)]
\item A \defin{Büchi} winning condition is of the form 
{$(V^*F)^\omega$}
%$\bigcap_{k\geq 0}V^kV^*FV^\omega$ 
for a set $F\subseteq V$ of final vertices, 
\ie winning plays are those that infinitely often visit vertices in $F$.
\item A \defin{co-B\"uchi} condition is of the form $V^{*}(V\setminus F)^{\omega}$  for a set $F\subseteq V$ of forbidden vertices, 
\ie winning plays are those that visit finitely often forbidden vertices.
\item A \defin{parity} winning condition is defined by a \emph{colouring} function $\col$ that is a mapping 
$\col: V \rightarrow \colors \subset \N$ where $\colors$ is a \emph{finite} set of \defin{colours}. The parity winning condition associated with $\col$ is the set 
$$\WC_\col = \{v_0 v_1 \cdots \in V^\omega \mid \liminf (\col(v_i))_{i \geq 0} \text{ is even}\}$$ \ie a play is winning if and only if the smallest colour infinitely often visited is even.
\end{itemize}

%Note that Büchi conditions are the parity conditions using colours $0$ and $1$.

Finally, a Büchi (\resp co-Büchi, parity) game is one equipped with a Büchi (\resp co-Büchi, parity) winning condition. For notation of such games we often replace the winning condition by the object that is used to defined it (\ie $F$ or $\col$).

%Of special interest are strategies that do not require memory. A
%\defin{positional strategy} is a strategy $\strat$ such that for any
%two partial plays of the form $\play.v$ and $\play'.v$, one has
%$\strat(\lambda\cdot v) = \strat(\lambda'\cdot v)$, \ie
%$\strat$ only depends on the current vertex.  It is a well known
%result that positional strategies suffice to win in parity
%games (see \eg \cite{Zie98} for a proof of this result in the setting of games played on infinite graphs).
%
%\begin{theorem}{(Positional determinacy)}\label{theo:posDet}
%Let $\game$ be a parity game. Then for any vertex, either \Eloise or \Abelard has a positional winning strategy.
%\end{theorem}
%
%We represent positional strategies as functions from $\VE$ (or $\VA$ depending on the player) into $V$.

\subsection{Tree Automata, Regular Tree Languages and Acceptance Game}\label{ssection:TARTLAG}

A \defin{tree automaton} $\A$ is a tuple $\langle A, Q , q_{ini},\Delta,\Acc \rangle$ where $A$ is the \defin{input alphabet}, $Q$ is the finite set of \defin{states}, $\qini\in Q$ is the \defin{initial state}, $\Delta \subseteq Q\times A \times Q \times Q$ is the \defin{transition relation} and $\Acc \subseteq Q^{\omega}$ is the \defin{acceptance condition}. An automaton is \defin{complete} if, for all $q\in Q$ and $a\in A$ there is at least one pair $(q_0,q_1)\in Q^2$ such that $(q,a,q_0,q_1)\in \Delta$. In this work we always assume that the automata are complete and this is implicit from now.

Given an $A$-labelled tree $\T$, a \defin{run} of $\A$ over $\T$ is a $Q$-labelled tree $\run$ such that 
\begin{enumerate}[(i)]
\item
 the root is labelled by the initial state, \emph{i.e.} $\run(\epsilon)=\qini $;
\item
 for all nodes $u$, $(\run(u),\T(u),\run(u\cdot0),\run(u\cdot1))\in\Delta$.
\end{enumerate}

A branch $\branch \in \Branches$ is \defin{accepting} in the run $\run$ if $\run(\branch)\in \Acc$, otherwise it is \defin{rejecting}. A run $\run$ is \defin{accepting} if \emph{all} its branches are accepting. Finally, a tree $\T$ is \defin{accepted} if there exists an accepting run of $\A$ over $\T$. The set of all trees accepted by $\A$ (or the language \emph{recognised} by $\A$) is denoted $L(\A)$. 

In this work we consider the following three classical acceptance conditions: 
%\acchanged{Büchi, co-Büchi and parity. These conditions are defined as in the game case but on sequences of states instead of sequences of vertices.}
 \begin{itemize}
 %\item A \defin{reachability} condition is given by a subset $F\subseteq Q$ of final states by letting $Reach(F) = Q^{*} FQ^{\omega}$, \emph{i.e.} a branch is accepting if it contains a final state.
 %\item A \defin{safety} condition is given by a subset $F\subseteq Q$ of forbidden states by letting $(Q\setminus F)^{\omega}$, \emph{i.e.} a branch is accepting if its labelling never contains a forbidden state.
 \item A \defin{B\"uchi} condition is given by a subset $F\subseteq Q$ of final states by letting $\Acc=Buchi(F) = (Q^{*}F)^{\omega}$, \emph{i.e.} a branch is accepting if it contains infinitely many final states.
 \item A \defin{co-B\"uchi} condition is given by a subset $F\subseteq Q$ of forbidden states by letting $coBuchi(F) = Q^{*}(Q\setminus F)^{\omega}$, \emph{i.e.} a branch is accepting it contains finitely many forbidden states.
 \item A \defin{parity} condition is given by a colouring mapping $\col: Q\rightarrow \nat$ by letting $$\Acc=Parity = \{q_0q_1q_2\cdots \mid \liminf (\col(q_i))_i \text{ is even}\}$$ \emph{i.e.} a branch is accepting if the smallest colour appearing infinitely often is even.
 \end{itemize}
They all are examples of $\omega$-regular acceptance conditions, \ie\ $\Acc$ is an $\omega$-regular set of $\omega$-words over alphabet $Q$ (see \eg \cite{PP04} for a reference book on infinite words languages). 

{
\begin{remark}
\label{rem:parity-eq-omega-reg}
The parity condition is expressive enough to capture the
  general case of an arbitrary $\omega$-regular
  condition $\Acc$. Indeed, it is well known that $\Acc$ is
    accepted by a \emph{deterministic} parity word automaton. By
    taking the synchronised product of this automaton with the tree
    automaton, we obtain a parity tree automaton accepting the same
    language (see \emph{e.g.} \cite{PP04}).
  \end{remark}
}

When it is clear from the context, we may replace, in the description of $\mathcal{A}$, $Acc$ by $F$ for Büchi/co-Büchi condition (\resp $\col$ for a parity condition), and we shall refer to the automaton as a Büchi/co-Büchi (\resp parity) tree automaton. A set $L$ of trees is an \defin{$\omega$-regular} tree language if there exists a parity tree automaton $\A$ such that $L=L(\A)$.   
%The class of $\omega$-regular tree languages is a Boolean algebra.
The class of $\omega$-regular tree languages is robust, as illustrated by the following famous statement.

\begin{theorem}{\cite{Rabin69}}\label{theo:regLang}
The class of $\omega$-regular tree languages is a Boolean algebra.
\end{theorem}

%The \defin{emptiness problem} is to decide for a given automaton $\A$ whether one has $L(\A)=\emptyset$. This problem is well-known to be polynomially equivalent to the problem of deciding the winner in a two-player parity game on a finite graph (and those colours used in the automaton and in the game are identical).
%
%

Fix an automaton $\A=\langle A, Q , \qini,\Delta,\Acc \rangle$ and a tree $t$ and define an acceptance game $\game_{\A,t}$, \ie a game where \Eloise wins if and only if there exists an accepting run of $\A$ on $t$, as follows. 

Intuitively, a play in $\game_{\A,t}$ consists in moving a pebble along a branch of
$t$ in a top-down manner: to the pebble is attached a state, and in a
node $u$ with state $q$ \Eloise picks a transition $(q,t(u),q_0,q_1) \in \Delta$, and then
{\Abelard} chooses to move down the pebble either to $u\cdot 0$ 
(and update the state to $q_0$) or to $u\cdot 1$ 
(and update the state to $q_1$). 

Formally (see Figure~\ref{fig:acceptance-game} for an illustration\footnote{In pictures, we always depict by circles (\resp squares) the vertices controlled by \Eloise (\resp \Abelard).}), let $G_{\A,t} = (V_\Ei \uplus V_\Ai,E)$ 
with $V_\Ei = Q\times \set{0,1}^*$,  $$V_\Ai = \set{(q,u,q_0,q_1) \mid u \in \set{0,1}^* \text{ and } (q,t(u),q_0,q_1) \in \Delta}\subseteq Q\times\set{0,1}^*\times Q\times Q$$ and
$$E = \set{((q,u),(q,u,q_0,q_1)) \mid (q,u,q_0,q_1) \in V_\Ai)}  \cup\set{((q,u,q_0,q_1),(u \cdot x,q_x)) \mid x \in \set{0,1} \text{ and } (q,u,q_0,q_1) \in V_\Ai)}$$
%$$\begin{array}{ll}
%E = & \set{((q,u),(q,u,q_0,q_1)) \mid (q,u,q_0,q_1) \in V_\Ai)} \quad  \cup \\
%& \set{((q,u,q_0,q_1),(u \cdot x,q_x)) \mid x \in \set{0,1} \text{ and } (q,u,q_0,q_1) \in V_\Ai)}\ 
%\end{array}$$
Then let $\arena_{\A,t} = (G_{\A,t},\VE,\VA)$ and 
extend $\col$ on $\VE\cup\VA$ by letting $\col((q,u))=\col((q,u,q_0,q_1))=\col(q)$. 
Finally define $\game_{\A,t}$ as the parity game $(\arena_{\A,t},\col)$. 

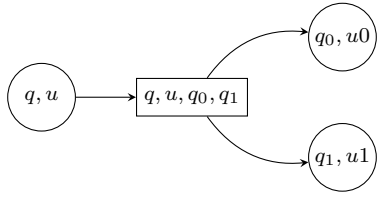
\begin{figure}
\begin{center}
\begin{tikzpicture}[scale=1,transform shape]
\tikzset{>=stealth}
  \node[draw,circle,inner sep=0mm,minimum size=9mm] (q) at (-2,0) {\footnotesize$q,u$};  
  \node[draw,inner sep=1mm,minimum size=5mm] (p) at (0,0) {\footnotesize$q,u,q_0,q_1$};  
  \node[draw,circle,inner sep=0mm,minimum size=9mm] (p0) at (2,0.8) {\footnotesize $q_{0},u0$};  
  \node[draw,circle,inner sep=0mm,minimum size=9mm] (p1) at (2,-0.8) {\footnotesize $q_{1},u1$};  

\draw[->] (q) --  (p);
\draw[->] (p) to [bend left] (p0);\draw[->] (p)  to [bend right] (p1);
\node (l) at (-0.5,-1.5) {\footnotesize for any $(q,t(u),q_{0},q_{1})\in\Delta$};  
\end{tikzpicture}
\caption{Local structure of the arena of the acceptance game $\game_{\mathcal{A},t}$.\label{fig:acceptance-game}}
\end{center}
\end{figure}

The next theorem is well-known (see \eg \cite{GurevichH82,GWT02}) and its proof is obtained by noting that strategies for \Eloise in $\game_{\A,t}$ are in bijection with runs of $\A$ on $t$.

\begin{theorem}{}
One has $t\in L(\A)$ if and only if \Eloise wins in $\game_{\A,t}$ from $(\qini,\epsilon)$.
\end{theorem}

%Now, in order to check emptiness for a given parity tree automaton $\A$, the idea is to slightly modify the previous game to additionally require \Eloise to describe a tree $t$ in $L(\A)$. Therefore, instead of picking a transition she also has to provide a node label. To make the resulting game finite one forget about the nodes name, \ie the component in $\set{0,1}^*$.
%
%Formally, one let $G_{\A} = (V'_\Ei \uplus V'_\Ai,E')$ 
%with $V'_\Ei = Q$ and $V'_\Ai = \Delta$ and
%$$\begin{array}{ll}
%E' = \qquad & \set{(q,(q,a,q_0,q_1)) \mid (q,a,q_0,q_1) \in \Delta)} \quad  \cup \\
%& \set{((q,a,q_0,q_1),q_x) \mid x \in \set{0,1} \text{ and } (q,a,q_0,q_1) \in \Delta)}\ 
%\end{array}$$
%Then one lets $\arena_{\A} = (G_{\A},\VE',\VA')$ and extends $\col$ on $\VE'\cup\VA'$ by letting $\col((q,a,q_0,q_1))=\col(q)$. 
%Finally, one defines $\game_{\A}$ as the parity game $(\arena_{\A},\col)$. 
%
%The next lemma is a well-known result and its proof is obtained by noting that strategies for \Eloise in $\game_{\A}$ are in bijection with pairs made of a tree $t$ and a run of $\A$ on $t$.
%
%\begin{lemma}{}
%For any automaton $\A$ one has $L(\A)\neq \emptyset$ if and only if \Eloise wins in $\game_{\A}$ from $\qini$.
%\end{lemma}
%
%

\section{Automata with Cardinality Constraints and Automata with Topological Bigness Constraints}\label{sec:autodef}

We now introduce the main notions studied in the paper, namely automata with cardinality constraints (studied in Section~\ref{sec:results-rejecting} and Section~\ref{sec:results-accepting}) and automata with topological bigness constraints (studied in Section~\ref{sec:results-topo}).

\subsection{Automata with Cardinality Constraints}

We now relax the criterion for a run to be accepting. Recall that classically, a run is accepting if \emph{every} branch in it is accepting. For a given automaton $\mathcal{A}$, we define the following four criteria (two for the case where one counts the number of accepting branches and two for the case where one counts the number of rejecting branches) for a run to be accepting.  Note that the case where one counts accepting branches was already considered in \cite{BeauquierNN91,BN95}.
\begin{itemize}
\item \textit{There are finitely many rejecting branches in the run.} A tree $t\in \LRejAtMostFin{\mathcal{A}}$ if and only if there is a run of $\mathcal{A}$ on $t$ satisfying  the previous condition.
\item \textit{There are at most a countably many rejecting branches in the run.} A tree $t\in \LRejAtMostCount{\mathcal{A}}$ if and only if there is a run of $\mathcal{A}$ on $t$ satisfying  the previous condition.
\item \textit{There are infinitely many accepting branches in the run.} A tree $t\in \LAccInf{\mathcal{A}}$ if and only if there is a run of $\mathcal{A}$ on $t$ satisfying  the previous condition.
\item \textit{There are uncountably many accepting branches in the run.} A tree $t\in \LAccUnc{\mathcal{A}}$ if and only if there is a run of $\mathcal{A}$ on $t$ satisfying the previous condition.
\end{itemize}

\subsection{Automata with Topological Bigness Constraints}

A notion of topological “bigness” and “smallness” is given by \emph{large} and \emph{meager} sets respectively (see \cite{VolzerV12,Graedel08} for a survey of the notion).
The idea is to see the set of branches in a tree as a topological space by taking as basic open sets the \emph{cones}. For a node $u\in\{0,1\}^*$, the cone $\mathrm{Cone}(u)$ is defined as $\{\pi\in\{0,1\}^\omega\mid u\prefix\pi\}$.
A set of branches $B\subseteq \{0,1\}^\omega$ is \defin{nowhere dense} if for all node $u$, there exists  $v\in\{0,1\}^*$ such that no branch of $B$ has $uv$ as a prefix.
It is \defin{meagre} if it is the countable union of nowhere dense sets. Finally it is \defin{large} if it is the complement of a meagre set.

For a given automaton $\mathcal{A}$, we define the following acceptance criterion: a run is accepting if and only if the set of accepting branches in it is large. Note that this is equivalent to require that the set of rejecting branches is meagre. 

Finally, a tree $t\in \LLarge{\mathcal{A}}$ if and only if there is a run of $\mathcal{A}$ on $t$ satisfying  the previous condition.

%%%%%%%%%%%% REJECTING %%%%%%%%%%%%
\section{Counting Rejecting Branches}\label{sec:results-rejecting}

For the classes of automata where acceptance is defined by a constraint on the number of rejecting branches we show that the associated languages are $\omega$-regular. For this, we adopt the following roadmap: first we design an acceptance game and then we note that it can be transformed into another \emph{equivalent} game that turns out to be the (usual) acceptance game for some tree automaton. 

{Fix, for this section, a parity tree automaton $\A=\langle A, Q , q_{ini},\Delta,\col \rangle$ and recall that a tree $t$ is in $\LRejAtMostCount{\A}$ (\resp in $\LRejAtMostFin{\A}$) if and only if there is a run of $\A$ on $t$ in which there are at most countably  (\resp finitely) many rejecting branches.}

%{Due to space constraints we mainly focus on the case $\LRejAtMostCount{\A}$, and we just give the construction and the result for the case  $\LRejAtMostFin{\A}$.}

\subsection{The Case of Languages $\LRejAtMostCount{\A}$}

Fix a tree $t$ and define an acceptance game for $\LRejAtMostCount{\A}$ as follows.
 In this game the two players move a pebble along a branch of $t$ in a top-down manner: to the pebble is attached a state whose colour gives the colour of the configuration. Hence, (\Eloise's main) configurations in the game are elements of $Q\times\{0,1\}^*$.
See Figure~\ref{fig:acceptance-game:RejCount} for the local structure of the arena.
 In a node $u$ with state $q$ \Eloise picks a transition $(q,t(u),q_0,q_1) \in \Delta$, and then
\Abelard has two possible options: \begin{enumerate}[(i)]
\item he chooses a direction $0$ or $1$; or
\item he lets \Eloise choose a direction $0$ or $1$.
\end{enumerate}
Once the direction $i\in\{0,1\}$ is chosen, the pebble is moved down to $u\cdot i$ and the state is updated to $q_i$. 
A play is won by \Eloise if one of the following two situations occurs: either the parity condition is satisfied or \Abelard has not let \Eloise infinitely often choose the direction.
Call this game $\gameRejCount{\A,t}$. 

The next theorem states that it is an acceptance game for language $\LRejAtMostCount{\A}$. 

\begin{figure}[tb]
\begin{center}
\begin{tikzpicture}[scale=1,transform shape]
\tikzset{>=stealth}
  \node[draw,circle,inner sep=0mm,minimum size=9mm] (q) at (-2,0) {\footnotesize$q,u$};  
  \node[draw,inner sep=1mm,minimum size=5mm] (p) at (0,0) {\footnotesize$q,u,q_0,q_1$};  
  \node[draw,circle,inner sep=0mm,minimum size=9mm] (p0) at (2,0.8) {\footnotesize $q_{0},u0$};  
  \node[draw,circle,inner sep=0mm,minimum size=9mm] (p1) at (2,-0.8) {\footnotesize $q_{1},u1$};  
  \node[draw,inner sep=1mm,minimum size=5mm,ellipse] (pE) at (4,0) {\footnotesize$q,u,q_0,q_1$};  

\draw[->] (q) --  (p);
\draw[->] (p) to [in=180,out=70] (p0);\draw[->] (p) to [in=180,out=-70] (p1);
\draw[->] (pE) to [in=0,out=110] (p0);\draw[->] (pE) to [in=0,out=-110] (p1);
\draw[->] (p) -- (pE);
\node (l) at (-1,-1.5) {\footnotesize for any $(q,t(u),q_{0},q_{1})\in\Delta$};  
\end{tikzpicture}
	\caption{Local structure of $\gameRejCount{\A,t}$.\label{fig:acceptance-game:RejCount}}
\end{center}
\end{figure}

\begin{theorem}
One has $t\in \LRejAtMostCount{\A}$ if and only if \Eloise wins in $\gameRejCount{\A,t}$ from $(\qini,\epsilon)$.
\end{theorem}

\begin{proof}
Assume that \Eloise has a winning strategy $\phi$ in $\gameRejCount{\A,t}$ from $(\qini,\epsilon)$. With $\phi$ we associate a run $\run$ of $\A$ on $t$ as follows. We inductively associate with any node $u$ a partial play $\play_u$ where \Eloise respects $\phi$. For this we let $\play_\epsilon = (\qini,\epsilon)$. Now assume that we have defined $\play_u$ for some node $u$ and let $(q,t(u),q_0,q_1)$ be the transition \Eloise plays from $\play_u$ when she respects $\phi$. Then let $i$ be the direction \Eloise would choose (again playing according to $\phi$) if \Abelard lets her pick the direction right after she played $(q,t(u),q_0,q_1)$: one defines $\play_{u\cdot i}$ as the partial  play obtained by extending $\play_u$ by \Eloise choosing transition $(q,t(u),q_0,q_1)$, followed by \Abelard letting her choose the direction and \Eloise choosing direction $i$; and one defines $\play_{u\cdot (1-i)}$ as the partial play obtained by extending $\play_u$ by \Eloise choosing transition $(q,t(u),q_0,q_1)$, followed by \Abelard choosing direction $(1-i)$. Note that for $j\in\{0,1\}$, $\play_{u\cdot j}$ ends with the pebble on $u\cdot j$ with state $q_j$ attached to it, equivalently in configuration $(q_j,uj)$. In the previous construction, we also refer to node $u\cdot i$ {(\ie the node that \Eloise has picked)} as \emph{marked}: note that any node has exactly one child that is marked (by convention the root is marked).

The run $\run$ is defined by letting $\run(u)$ be the state attached to the pebble in the last configuration of $\play_u$. By construction, $\run$ is a valid run of $\A$ on $t$ and moreover with any branch $\pi$ in $\run$ one can associate a play $\play_\pi$ in $\gameRejCount{\A,t}$ from $(\qini,\epsilon)$ where \Eloise respects $\phi$ (one simply considers the limit of the \emph{increasing} sequence of partial plays $\play_u$ where $u$ ranges over those nodes along branch $\pi$). By construction $\pi$ is rejecting if and only if $\play_\pi$ does not fulfil the parity condition. 

Now consider a rejecting branch $\pi$. As $\play_\pi$ does not fulfil the parity condition and as $\phi$ is winning so does $\play_\pi$ hence, it means that in $\play_\pi$ \Abelard does not let \Eloise choose infinitely often the direction. Equivalently, $\pi$ contains finitely many marked nodes (marked nodes corresponding precisely to those steps where \Eloise chooses the direction). Hence, with any rejecting branch $\pi$, one can associate the last marked node $u_\pi$ in it. And if $\pi\neq \pi'$ one has $u_\pi\neq u_{\pi'}$: indeed, at the point where $\pi$ and $\pi'$ first differs one of the node is marked from the property that every node has exactly one child that is marked. Hence, the number of rejecting branches is countable as the map $\pi\mapsto u_\pi$ is injective and as the number of nodes in a tree is countable. This permits to conclude that $\run$ is an accepting run –~in the sense of  $\LRejAtMostCount{\A}$~– of $\A$ on $t$.

%\os{J'ai changé un peu ta modif Axel; en particulier le théorème de Martin est dans un cadre bien plus général que celui des jeux $\omega$-régulier. Si on a besoin de gagner de la place, je propose qu'on revienne à la version précédente sans citer le papier de Martin…}
{Conversely, assume that \Eloise has no winning strategy. It follows from Borel determinacy \cite{Martin75} that  \Abelard has a winning strategy $\psi$ in $\gameRejCount{\A,t}$ from $(\qini,\epsilon)$. Let us prove that any run $\run$ of $\A$ on $t$ contains  uncountably many rejecting branches.} For this, fix a run $\run$ of $\A$ on $t$.
%By contradiction, assume that there is a run $\run$ of $\A$ on $t$ that has at most countably many rejecting branches. 
With any sequence $\alpha=\alpha_1\alpha_2\cdots \in\{0,1\}^\omega$ we associate a strategy $\phi_\alpha$ of \Eloise in $\gameRejCount{\A,t}$. The strategy $\phi_\alpha$ of \Eloise consists in describing the run $\run$ and to propose direction $\alpha_i$ when it is the $i$-th time that \Abelard lets her choose the direction. More formally, when the pebble is on node $u$ with state $q$ (we will trivially have $q=\run(u)$ as an invariant) she picks the transition $(\run(u),t(u),\run(u0),\run(u1))$; moreover if \Abelard lets her choose the direction, she picks $\alpha_{i+1}$ where $i$ is the number of time \Abelard let her choose the direction from the beginning of the play.

As we assumed that $\psi$ is winning, the (unique) play obtained when she plays $\phi_\alpha$ and when he plays $\psi$ is loosing for \Eloise: such a play defines a branch $\pi_\alpha$ in $\run$, and this branch is a rejecting one. Now, for any $\alpha\neq \alpha'$ one has $\pi_\alpha\neq \pi_{\alpha'}$: indeed, at some points $\alpha$ and $\alpha'$ differs and, as infinitely often \Abelard lets \Eloise chooses the directions, the branches $\pi_\alpha$ and $\pi_{\alpha'}$ will differ as well. But as there are uncountably many different sequences $\alpha$, it leads an uncountable number of rejecting branches in $\run$. Hence, $\run$ is rejecting.
\end{proof}

Consider game $\gameRejCount{\A,t}$ and modify it so that \Eloise is now announcing in advance which direction she would choose if \Abelard let her do so. This new game is equivalent to the previous one (meaning that she has a winning strategy in one game if and only if she also has one in the other game). As this new game can easily be modified to obtain an \emph{equivalent} acceptance game for the classical acceptance condition (as described in Section~\ref{ssection:TARTLAG}) one concludes that the languages of the form $\LRejAtMostCount{\A}$ are always $\omega$-regular.

\begin{theorem}
Let $\A=\langle A, Q , q_{ini},\Delta,\col \rangle$ be a parity tree automaton using $d$ colours. Then there exists a parity tree automaton $\A'=\langle A, Q' , q'_{ini},\Delta',\col' \rangle$ such that $ \LRejAtMostCount{\A}=L(\A')$. Moreover $|Q'|=\mathcal{O}(d|Q|)$ and $\A'$ uses $d+1$ colours.
\end{theorem}

\begin{proof}
Define $\gameRejCountB{\A,t}$ as the game obtained from $\gameRejCount{\A,t}$ by asking \Eloise to say which direction she would choose before \Abelard possibly lets her this option. \Eloise has a winning strategy in $\gameRejCount{\A,t}$ if and only if she has a winning strategy in $\gameRejCountB{\A,t}$ (strategies being essentially the same in both games). The way she indicates the direction can be encoded in the control state: just duplicate the control states (with a classical version and a starred version of each state) and when she wants to pick transition \eg $(q,t(u),q_0,q_1)$ and direction $1$, she just moves to configuration $(q,u,q_0,q_1^*)$ in the new game. Now the winning condition can be rephrased as either the parity condition is satisfied or finitely many configuration of the form $(q^*,u)$ are visited. Now this later game can be transformed into a standard acceptance game for $\omega$-regular language (as defined in Section~\ref{ssection:TARTLAG}) by the following trick. One adds to states an integer where one stores the smallest colour seen since the last starred state was visited (this colour is easily updated); whenever a starred state is visited the colour is reseted to the colour of the state. Now unstarred states are given an even colour that is greater than all colour previously used (hence, it ensures that if finitely many starred states are visited \Eloise wins) and starred states are given the colour that was stored (hence, if infinitely many starred states are visited we retrieve the previous parity condition).
It should then be clear that the later game is a classical acceptance game, showing that $\LRejAtMostCount{\A}$ is $\omega$-regular. 

The construction of $\A'$ is immediate from the final game and the size is linear in $d|Q|$ due to the fact that one needs to compute the smallest colour visited between two starred states.
\end{proof}

\subsection{The Case of Languages $\LRejAtMostFin{\A}$}

%The key idea is to note that a set of branches is finite if and only if there are finitely many nodes belonging to at least two branches in the set.

The following lemma (whose proof is straightforward) characterises \emph{finite} sets of branches by noting that for such a set there is a finite number of nodes belonging to at least two branches in the set.

\begin{lemma}\label{lemma:finite}
Let $\Pi$ be a set of branches. Then $\Pi$ is finite if and only if the set $W=\{u\in \{0,1\}^*\mid\exists \pi_0\neq\pi_1\in\Pi \text{ s.t. } u\prefix\pi_0  \text{ and } u\prefix\pi_1\}$ is finite.
Equivalently, $\Pi$ is finite if and only if there exists some $\ell\geq 0$ such that for all $u\in \{0,1\}^{\geq \ell}$ there is at most one $\pi\in\Pi$ such that $u\prefix \pi$.
\end{lemma}

Now, fix a tree $t$ and define an acceptance game for $\LRejAtMostFin{\A}$ as follows. 
%In this game the two players move a pebble along a branch of $t$ in a top-down manner: the pebble is attached a state and the colour of the state gives the colour of the configuration.
In this game (we refer the reader to Figure~\ref{fig:acceptance-game:RejAtMostFin} for the local structure of the arena for game $\gameRejFin{\A,t}$) the two players move a pebble along a branch of $t$ in a top-down manner: as in the classical case the players first select a transition and then a direction. The colour of the current state gives the colour of the configuration.
There are three modes in this game: \emph{wait} mode, \emph{path} mode and \emph{check} mode and the game starts in \emph{wait} mode. Hence, 
%(main) 
{(\Eloise's main)}
configurations in the game are elements of $Q\times\{0,1\}^*\times \{wait,path,check\}$.

\begin{figure}
\begin{center}
\begin{tikzpicture}[scale=.9,transform shape]
\tikzset{>=stealth}

%%%%% wait mode%%%%%
\begin{scope}
%%% Noeuds %%%
  \node[draw,circle,inner sep=0mm,minimum size=9mm] (Wq) at (0,0.5) {\footnotesize$q,u$};  
  \node[draw,inner sep=1mm,minimum size=5mm] (WqTwc) at (3,2) {\footnotesize$q,u,q_0^w,q_1^c,0$};  
  \node[draw,inner sep=1mm,minimum size=5mm] (WqTcw) at (3,-1) {\footnotesize$q,u,q_0^c,q_1^w,0$};  
  \node[draw,inner sep=1mm,minimum size=5mm] (WqTww) at (2,0.5) {\footnotesize$q,u,q_0^w,q_1^w,0$};  
  \node[draw,inner sep=1mm,minimum size=5mm] (WqTcp) at (3,-4) {\footnotesize$q,u,q_0^c,q_1^p,0$};  
  \node[draw,inner sep=1mm,minimum size=5mm] (WqTpc) at (3,-3) {\footnotesize$q,u,q_0^p,q_1^c,0$};  
  \node[draw,inner sep=1mm,minimum size=5mm] (WqTcc) at (3,-2) {\footnotesize$q,u,q_0^c,q_1^c,0$};  
  \node[draw,circle,inner sep=0mm,minimum size=9mm] (Wq0) at (4,1.2) {\footnotesize $q_{0},u0$};  
  \node[draw,circle,inner sep=0mm,minimum size=9mm] (Wq1) at (4,-.2) {\footnotesize $q_{1},u1$};  
%%%% Arcs
\draw[->] (Wq) |-  (WqTwc);
\draw[->] (Wq) --  (WqTww);
\draw[->] (Wq) to [bend right]  (WqTcw);
\draw[->] (Wq) to [bend right]  (WqTcc);
\draw[->] (Wq) to [bend right]  (WqTpc);
\draw[->] (Wq) |-  (WqTcp);
%%%%
\draw[->] (WqTww) to [bend left]  (Wq0);
\draw[->] (WqTww) to [bend right]  (Wq1);
\draw[->] (WqTwc) to [out=0,in=30]  (Wq0);
\draw[->] (WqTcw) to [out=0,in=-30]  (Wq1);
\node[rotate=90] (l) at (-1,-0.5) {\footnotesize for any $(q,t(u),q_{0},q_{1})\in\Delta$};  
\node (path) at (2,3) {\footnotesize \emph{\underline{wait mode}}};  
\draw[dotted] (5,2.5) -- (5,-4.5);
\end{scope}

%%%%% check mode%%%%%
\begin{scope}[xshift=6.5cm]
%%% Noeuds %%%
  \node[draw,circle,inner sep=0mm,minimum size=9mm] (Cq) at (2,-1.5) {\footnotesize$q,u$};  
  \node[draw,inner sep=1mm,minimum size=5mm] (CqTcc) at (0,-1.5) {\footnotesize$q,u,q_0^c,q_1^c,0$};  
  \node[draw,circle,inner sep=0mm,minimum size=9mm] (Cq0) at (0,-3.5) {\footnotesize $q_{0},u0$};  
  \node[draw,circle,inner sep=0mm,minimum size=9mm] (Cq1) at (0,.5) {\footnotesize $q_{1},u1$};  
%%%% Arcs
\draw[->] (Cq) --  (CqTcc);
%%%%
\draw[->] (CqTcc) to  (Cq0);
\draw[->] (CqTcc) to  (Cq1);
\node (l) at (0.75,1.5) {\footnotesize for any $(q,t(u),q_{0},q_{1})\in\Delta$};  
\node (path) at (1,3) {\footnotesize \emph{\underline{check mode}}};  
\draw[dotted] (3,2.5) -- (3,-4.5);
\end{scope}
%%%%%
%%Arcs entre wait et check
\draw[->] (WqTwc) to [out=10,in=165] (Cq1);
\draw[->] (WqTcw) to [out=-10,in=165] (Cq0);
\draw[->] (WqTpc) to [out=0,in=200] (Cq1);
\draw[->] (WqTcp) to [out=0,in=200] (Cq0);
\draw[->] (WqTcc) to [out=-10,in=180] (Cq0);
\draw[->] (WqTcc) to [out=10,in=180] (Cq1);

%%%%% path mode%%%%%
\begin{scope}[xshift=11cm]
%%% Noeuds %%%
  \node[draw,circle,inner sep=0mm,minimum size=9mm] (Pq) at (2,-1.5) {\footnotesize$q,u$};  
  \node[draw,inner sep=1mm,minimum size=5mm] (PqTpc) at (0,-1) {\footnotesize$q,u,q_0^p,q_1^c,0$};  
  \node[draw,inner sep=1mm,minimum size=5mm] (PqTcp) at (0,-2) {\footnotesize$q,u,q_0^c,q_1^p,0$};  
  \node[draw,circle,inner sep=0mm,minimum size=9mm] (Pq0) at (0,.5) {\footnotesize $q_{0},u0$};  
  \node[draw,circle,inner sep=0mm,minimum size=9mm] (Pq1) at (0,-3.5) {\footnotesize $q_{1},u1$};  
%%%% Arcs
\draw[->] (Pq) to [out=170,in=0]  (PqTpc);
\draw[->] (Pq) to [out=190,in=0] (PqTcp);
%%%%
\draw[->] (PqTpc) to [out=90,in=270]  (Pq0);
\draw[->] (PqTcp) to [out=270,in=90]  (Pq1);

\node (l) at (0.75,1.5) {\footnotesize for any $(q,t(u),q_{0},q_{1})\in\Delta$};  
\node (path) at (0.5,3) {\footnotesize \emph{\underline{path mode}}};  

\end{scope}

%%%%%
%%Arcs entre wait et path
\draw[->] (WqTpc) to [out=-5,in=235] (Pq0);
\draw[->] (WqTcp) to [out=-10,in=200] (Pq1);
%%Arcs entre path et check
\draw[->] (PqTpc) to [out=100,in=0]  (Cq1);
\draw[->] (PqTcp) to [out=260,in=0]  (Cq0);

\end{tikzpicture}
\caption{Local structure of the arena of the acceptance game $\gameRejFin{\A,t}$. We use superscript to indicate which modes have been proposed by \Eloise.\label{fig:acceptance-game:RejAtMostFin}}
\end{center}
\end{figure}
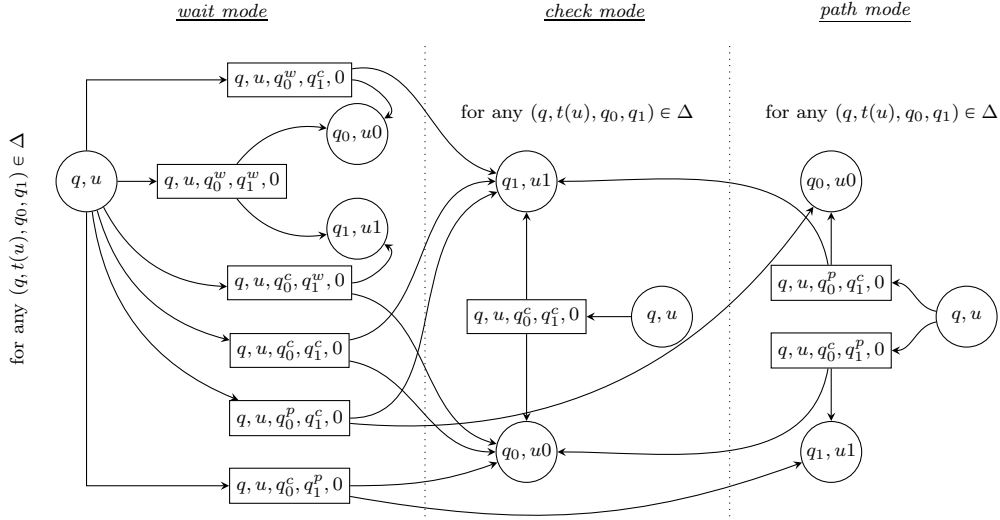

Regardless of the mode, in a node $u$ with state $q$ \Eloise picks a transition $(q,t(u),q_0,q_1) \in \Delta$, and for each direction in $i\in\{0,1\}$ she proposes the next mode $m_i$ in $\{wait,path,check\}$ (we describe below what are the possible options depending on the current mode). Then \Abelard chooses a direction $j\in\{0,1\}$, the pebble is moved down to $u\cdot j$, the state is updated to $q_{j}$ and the mode changes to $m_j$. 
The possible modes that \Eloise can propose depend on the current mode in the following manner.
\begin{itemize}
\item In \emph{wait} mode she can propose any modes $m_i$ in $\{wait,path,check\}$ but if one proposed mode $m_i$ is \emph{path} then the other mode $m_{1-i}$ must be \emph{check}.
\item In \emph{check} mode the proposed modes must be \emph{check} (\ie once the mode is \emph{check} it no longer changes).
\item In \emph{path} mode one proposed mode must be \emph{path} and the other must be \emph{check}.
\end{itemize}

A play is won by \Eloise if one of the two following situation occurs.
\begin{itemize}
\item The \emph{wait} mode is eventually left and the parity condition is satisfied.
\item The mode is eventually always equal to \emph{path}. 
\end{itemize}

In particular a play in which the mode is \emph{wait} forever is lost by \Eloise.
{Note that the latter winning condition can easily be reformulated as a parity condition.}
Call this game $\gameRejFin{\A,t}$. 

The next theorem states that it is an acceptance game for language $\LRejAtMostFin{\A}$.

\begin{theorem}\label{theo:RejAtMostFin_AcceptanceGame}
One has $t\in \LRejAtMostFin{\A}$ if and only if \Eloise wins in $\gameRejFin{\A,t}$ from $(\qini,\epsilon,wait)$.
\end{theorem}

\begin{proof}
Assume that \Eloise has a winning strategy $\phi$ in $\gameRejFin{\A,t}$ from $(\qini,\epsilon,wait)$.
With $\phi$ we associate a run $\run$ of $\A$ on $t$ as follows. We inductively define for any node $u\in \{0,1\}^*$ a partial play $\play_u$ where \Eloise respects $\phi$. For this we let $\play_\epsilon = (\qini,\epsilon,wait)$. Now assume that we defined $\play_u$ for some node $u\in \{0,1\}^*$ and let $(q,t(u),q_0,q_1)$ be the transition \Eloise plays from $\play_u$ when she respects $\phi$. Then for each $i\in\{0,1\}$ one defines $\play_{u\cdot i}$ as the partial play obtained by extending $\play_u$ by \Eloise choosing transition $(q,t(u),q_0,q_1)$, followed by \Abelard choosing direction $i$ (we update the mode accordingly to the choice of \Eloise when respecting $\phi$ in $\play_u$).

The run $\run$ is defined by letting, for any $u\in \{0,1\}^*$, $\run(u)$ be the state attached to the pebble in the last configuration of $\play_u$. By construction, $\run$ is a run of $\A$ on $t$. Moreover with any branch $\pi$ one can associate a play $\play_\pi$ in $\gameRejFin{\A,t}$ from $(\qini,\epsilon)$ where \Eloise respects $\phi$ (one simply considers the limit of the \emph{increasing} sequence of partial plays $\play_u$ where $u$ ranges over nodes along branch $\pi$). 

First, note that there exists some $\ell\geq 0$ such that, for all $u\in\{0,1\}^{\geq \ell}$, $\play_u$ ends in a vertex where the mode is not \emph{wait}. Indeed, if this was not the case, one could construct an infinite branch $\pi$ such that, for all node $u$ in $\pi$, $\play_u$ ends with a vertex in mode \emph{wait} (recall that the only way to be in \emph{wait} mode is to be in that mode from the very beginning) and therefore the corresponding play $\play_\pi$ would be loosing, which contradicts the fact that \Eloise respects her winning strategy $\phi$ in play $\play_\pi$. Now, consider some node $u\in\{0,1\}^\ell$. If the final vertex in $\play_u$ is in mode \emph{check} one easily verifies that any branch that goes through $u$ is accepting (because the corresponding play is winning hence, satisfies the parity condition). If the final vertex in $\play_u$ is in mode \emph{path} one easily checks that among all branches that goes through $u$, there is exactly one branch $\pi$ such that $\play_\pi$ eventually stays in mode \emph{path} forever (and this branch may not satisfy the parity condition) while all other branches eventually stay in mode \emph{check} forever (and satisfy the parity condition). Therefore, the number of rejecting branches is finite. 

Conversely, assume that there is a run $\run$ of $\mathcal{A}$ on $t$ that contains finitely many rejecting branches. Call $\Pi$ this set of branches. Thanks to Lemma~\ref{lemma:finite},  there exists some $\ell\geq 0$ such that for all $w\in \{0,1\}^{\geq \ell}$ there is at most one $\pi\in\Pi$ such that $w\prefix \pi$.
Using $\run$ we define a strategy $\phi$ for \Eloise in  $\gameRejFin{\A,t}$ as follows. 
%The strategy $\phi$ uses as a memory a node $u\in\{0,1\}^*$ and initially this memory is set to be $\epsilon$. 
In any configuration $(q,u)$ (regardless of the mode) the strategy is to play transition $(q,t(u),\run(u 0),\run(u 1))$. Then there are several cases for determining how the mode is updated.
\begin{itemize}
\item  In some configuration $(q,u)$ with $u$ of length strictly smaller than $\ell$ the mode remains in \emph{wait}. 
\item In some configuration $(q,u)$ with $u$ of length equal to $\ell$ the strategy  proposes to update the mode to \emph{path} for direction $i\in\{0,1\}$ such that $u\cdot i\prefix \pi$ for some branch $\pi\in\Pi$, and to \emph{check} otherwise. Note that due to the definition of $\ell$, there is at most one direction $i$ in which the mode becomes \emph{path}.
\item In some configuration $(q,u)$ with $u$ of length strictly greater than $\ell$ if the mode is \emph{check} it will remain to \emph{check} in both direction. Otherwise (\ie the mode is \emph{path}) the strategy  proposes to update the mode to \emph{path}  for direction $i\in\{0,1\}$ such $u\cdot i\prefix \pi$ for some branch $\pi\in\Pi$, and to \emph{check} otherwise. Note that in the latter case, there is exactly one direction $i$ in which the mode is \emph{path}.
\end{itemize}
Remark that no play where \Eloise respects $\phi$ stays in \emph{wait} mode forever. Moreover, with any $\play$ where \Eloise respects $\phi$ one can associate a branch in the run $\run$ and this branch is rejecting if and only if $\play$ stays eventually in mode \emph{path} forever. Hence, any play where the mode is not infinitely often \emph{path} satisfies the parity condition (because the corresponding branch in $\run$ does so). Hence, $\phi$ is winning.
\end{proof}

From Theorem~\ref{theo:RejAtMostFin_AcceptanceGame} and the local structure of the arena of game $\gameRejFin{\A,t}$ one easily concludes that any language of the form $\LRejAtMostFin{\A}$ is $\omega$-regular.

\begin{theorem}\label{theo:RejAtMostFin_regular}
Let $\A=\langle A, Q , q_{ini},\Delta,\col \rangle$ be a parity tree automaton using $d$ colours. Then, there exists a parity tree automaton $\A'=\langle A, Q' , q'_{ini},\Delta',\col' \rangle$ such that $\LRejAtMostFin{\A}=L(\A')$. Moreover, $|Q'|=\mathcal{O}(|Q|)$ and $\A'$ uses $d$ colours.
\end{theorem}

\begin{proof}
Consider the local structure of game $\gameRejFin{\A,t}$ as described in Figure~\ref{fig:acceptance-game:RejAtMostFin}. It is then fairly simple how one defines $\A'$: for any state on $\A$ and any mode, one gets a new state in $\A'$, and the transition function $\Delta'$ directly follows from the way we update the modes. The states in \emph{wait} mode all get the same odd minimal colour (hence, if they are never left \Eloise looses), the states in path mode all get the same even minimal colour (hence, if they are never left \Eloise wins), and the states in check mode get the colour they had in $\A$.
Hence, we do not need to add any extra colour (except in the case where $\A$ uses only one colour but in this very degenerated case one can simply take $\A'=\A$).
\end{proof}

\subsection{Languages $\LRejAtMostCount{\A}$ and $\LRejAtMostFin{\A}$ \emph{vs} Büchi Tree Languages}

One can wonder, as it will be later the case (see Section~\ref{sec:results-accepting}) for languages of the form $\LAccInf{\A}$ or $\LAccUnc{\A}$ , whether a Büchi condition is enough to accept (with the classical semantics) a language of the form $\LRejAtMostCount{\A}$ (\resp $\LRejAtMostFin{\A}$). The next Proposition answers negatively.

\os{J'ai relu la preuve mais pas de près}
\ac{J'ai fait des modifications. Il y avait des petits bugs dans les détails.}
\begin{proposition}\label{proposition:noBuchi}
There is a co-Büchi deterministic tree automaton $\A$ such that for any \emph{Büchi} tree automaton $\A'$,  $ \LRejAtMostCount{\A}\neq L(\A')$ and $ \LRejAtMostFin{\A}\neq L(\A')$.
\end{proposition}

\begin{proof}
We choose for $\A$ the same automaton that was used by Rabin in \cite{Rabin72} to derive a similar statement where one replaces $\LRejAtMostCount{\A}$ by $L(\A)$ and we generalise the proof of this result as given in \cite[Example~6.3]{Thomas97}.

Let $L$ be the set of $\{a,b\}$-labeled trees such that the number of branches that contain infinitely many $b$'s is at most countable. Obviously there is a \emph{deterministic} co-Büchi automaton $\A$ such that $L=\LRejAtMostCount{\A}$. Indeed, consider an automaton $\A$ with two states, one forbidden and the other one non-forbidden, and that from any state, goes (for both sons) in the forbidden state whenever he was in a $b$-labeled node and otherwise goes (for both sons) in the non-forbidden state.

Assume, by contradiction, that there is some B\"uchi tree automaton $\A'=\langle \{a,b\}, Q, q_{ini},\Delta,{F} \rangle$ such that $ \LRejAtMostCount{\A}= L(\A')$. {Note that we will not treat the case where $ \LRejAtMostFin{\A}= L(\A')$ as it is identical}. Let \acchanged{$n = |Q|$} and let $t$ be the $\{a,b\}$-labeled tree such that $t(u)=b$ if and only if $u\in (1^+0)^k$ for some $1\leq k\leq n$, \ie label $b$ occurs when a left successor is taken after a sequence of right successors, however allowing at most $n$ left turns. Clearly, $t\in L$ as every branch contains finitely many $b$-labeled nodes. Let $\run$ be an accepting run of $\A'$ on $t$.

The goal is to exhibit three nodes $u$, $uv$ and $uw$ such that:
\begin{enumerate}
\item $uv$ is not a prefix of $uw$ and \emph{vice versa};
\item $\run(u)=\run(uv)=\run(uw)$ is a final state;
\acchanged{
\item $t(u)=t(uv)=t(uw)=a$;
}
\item on the path segment \acchanged{from $u$ to $uv$}  there is at least one node labeled $b$;
\item on the path segment \acchanged{from $u$ to $uw$}  there is at least one node labeled $b$.
\end{enumerate}

Once this is done we can form a new tree $t'$ (and an associated run $\run'$ of $\A'$ on $t'$) by iterating the finite path segment from $u$ (\acchanged{inclusive}) to $uv$ (\acchanged{exclusive}) and from $u$ (\acchanged{inclusive}) to $uw$ (\acchanged{exclusive}) indefinitely, copying also the subtrees which have their roots on these path segments. \acchanged{More formally, consider the two-hole context $C_t[\bullet,\bullet]$ (\resp $C_\rho[\bullet,\bullet]$) obtained by placing holes at
$uv$ and $uw$ in $t$ (\resp in $\rho$) and the two-hole context $D_t[\bullet,\bullet]$ (\resp $D_\rho[\bullet,\bullet]$) obtained by placing holes at $v$ and $w$ in the subtree $t_{/u}$ of $t$ rooted at $u$  (\resp $\rho_{/u}$) . The tree $t'$ is equal to $C_t[t'',t'']$ where $t''$ is the unique tree satisfying  the equation $t''=D_t[t'',t'']$. Similarly $\rho'$ is equal is equal to $C_\rho[\rho'',\rho'']$ where $\rho''$ is the unique tree satisfying  the equation $\rho''=D_\rho[\rho'',\rho'']$.}

This process leads to naturally exhibit a binary tree like structure\footnote{Actually it is what we refer to as an \emph{accepting-pseudo binary tree} in Section~\ref{section:PBT}.} inside $t'$/$\run'$ such that any branch in $t'$ contains infinitely many $b$-labeled nodes, hence $t'\notin L$ (indeed, there will be uncountably many branches in this binary tree like structure). But this leads a contradiction as $\run'$ is easily seen to be accepting while being a run on $t'$.

We now explain why we can find nodes $u$, $uv$ and $uw$ as claimed above.

For all $k$, we denote by $L_{k}$ the set $1^{+}(0 1^{+})^{k}$ of nodes 
that can be reached by making $k$ left-turns. We write $L_{<k}$ for the union
$\bigcup_{0\leq i<k} L_{i}$ and we write $L_{\geq k}$ for the union $\bigcup_{i\geq k} L_{i}$.

For all $k<n$ and all node $u \in L_{k}$, we denote by $W(u)$ the set of 
final states $q \in F$ such that there exist $v$ and \acchanged{$w \in L_{>k} \cap L_{\leq n} $} where $v$ is not prefix of $w$ and $w$ is not a prefix of $v$ and $\rho(uv)=\rho(uw)=q$. The set $W(u)$ is not empty as the run $\run$ is accepting. Moreover, for $u \in L_{k}$ and $u' \in L_{k'}$ with $k\leq k'<n$, if $u \prefix u'$ then $W(u) \supseteq W(u')$.

Clearly, it is enough to find a $u$ in $L_{<n}$ such that
$\rho(u) \in F$ and $\rho(u) \in W(u)$. 

We construct by induction an increasing sequence  of nodes $u_{0},\ldots,u_{n-1}$
such that:
\begin{itemize}
\item for all $i \in [0,n-1]$, $u_{i} \in L_{i}$,
\item $\rho(u_{0}) \in F$,
\item for all $i \in [1,n-1]$, $\rho(u_{i}) \in W(u_{i-1})$.
\end{itemize}

For $u_{0}$, we pick any node in $1^{*}$ labeled by a final state. Such a node must exist as the branch $1^{\omega}$ is accepted by the automaton.

Assume that $u_{i}$ has been constructed, we pick for $u_{i+1}$ a node 
in $u_{i} 1^{+}01^{+}$ labeled by a state in $W(u_{i})$. As all branches
of the form $u_{i} 1^{m} 0 1^{\omega}$ are accepted,  there exists a final state
$q$ such that $\rho(u_{i} 1^{m_{1}} 0 1^{m_{1}'})=\rho(u_{i} 1^{m_{2}} 0 1^{m_{2}'})=q$ for $m_{1} \neq m_{2}$. This implies that $q$ belongs to $W(u_{i})$
and we take for instance $u_{i+1}=u_{i} 1^{m_{1}} 0 1^{m_{1}'}$.

There must exist $i<n$ such that $\rho(u_{i})$ belongs to $W({u_{i}})$. 
Assume toward a contradiction that it is not the case. Then we have $n>|W(u_{0})|$ as $\rho(u_{0})$ does not belong to $W(u_{0})$. For all $0 \leq i <n$, $W(u_{i}) \supsetneq W(u_{i+1})$ as $W(u_{i}) \supseteq W(u_{i+1})$ and  $q_{i+1}$ belongs to $W(u_{i})$ but not to $W(u_{i+1})$.
Hence, we have:
\[
n>|W(u_{0})|> \cdots > |W(u_{n-1})|
\]
It follows that $W(u_{n-1})= \emptyset$ which brings the contradiction.
\end{proof}

%%%%%%%%%%%% ACCEPTING %%%%%%%%%%%%

\section{Counting Accepting Branches}\label{sec:results-accepting}

We now consider the case where acceptance is defined by a constraint on the number of accepting branches and we show that the associated languages are $\omega$-regular. It leads to new proofs, that rely on games rather than on logic, of the results in \cite{BN95}.

{Fix, for this section, a parity tree automaton $\A=\langle A, Q , q_{ini},\Delta,\col \rangle$ and recall that a tree $t$ is in $\LAccInf{\A}$ (\resp $\LAccUnc{\A}$) if and only if there is a run of $\A$ on $t$ that contains infinitely (\resp uncountably) many accepting branches.}

\subsection{The Case of Languages $\LAccInf{\A}$}

The key idea behind defining an acceptance game for $\LAccInf{\A}$ for some tree $t$ is to exhibit a \emph{pseudo comb} in a run of $\A$ over $t$. In a nutshell,
a pseudo comb consists of an infinite branch $U$ and a collection $V$ of accepting branches each of them sharing some prefix with $U$. One easily proves that a run contains infinitely many accepting branches if and only if it contains a pseudo comb.

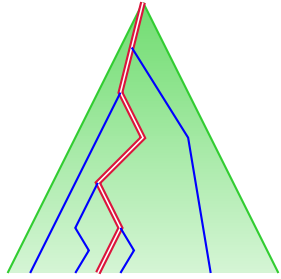
\begin{figure}[htb]
\begin{center}
\begin{tikzpicture}

\begin{scope}[scale = 0.6]
\fill[bottom color=LimeGreen!100!black!20,
top color=LimeGreen!100!black!70] (0,0) -- (-3,-6) -- (3,-6) -- cycle ;
\draw[LimeGreen,thick] (0,0) -- (-3,-6);
\draw[LimeGreen,thick] (0,0) -- (3,-6);
\draw[Crimson,double,thick] (0,0) -- (-0.5,-2) -- (0,-3) -- (-1,-4) -- (-.5,-5) -- (-1,-6);
%		\draw[Crimson,double,thick] (0,0) -- (-0.5,-2) -- (0,-3) -- (-1,-4) -- (-.5,-5) -- (-1,-6);
%\draw[Crimson,thick] (0,0) -- (-0.5,-1.2) -- (0,-3) -- (-1.5,-4) -- (-1,-5) -- (-2.5,-6);
\draw[Blue,thick] (-0.25,-1) -- (1,-3) -- (1.5,-6);
\draw[Blue,thick] (-0.5,-2) -- (-2.5,-6);
\draw[Blue,thick] (-1,-4) -- (-1.5,-5) -- (-1.2,-5.5) -- (-1.5,-6);
\draw[Blue,thick] (-.5,-5) -- (-.2,-5.5) -- (-.5,-6);
\end{scope}

\end{tikzpicture}
\caption{A pseudo comb $(\textcolor{Crimson}{\mathbb{U}},\textcolor{Blue}{V})$}\label{sfigure:Comb}
\end{center}
\end{figure}

More formally, a \defin{pseudo comb} (see Figure~\ref{sfigure:Comb} for an illustration) is a pair of subset $(U,V)$ of nodes with $U,V\subseteq \{0,1\}^*$ such that:
\begin{itemize}
\item $U$ and $V$ are disjoint.
\item $U$ is a branch: $\epsilon\in U$ and for all $u\in U$ one has $|\{u0,u1\}\cap U|=1$.
%\item $V=\bigcup_{i\geq 0}V_i$ is a collection $(V_i)_{i\geq 0}$ of branches: for all $i\geq 0$, for all $v\in V_i$ one has $|\{v0,v1\}\cap V_i|=1$
\item $V$ is a set of nodes such that \begin{enumerate}[(i)]
	\item for all $v\in V$ one has $|\{v0,v1\}\cap V|=1$;
	\item for all $v\in V$, $v\in (U\cup V)\cdot\{0,1\}$.
\end{enumerate}
\item For infinitely many $u\in U$ there exists some $v\in V$ such that either $v=u0$ or $v=u1$.
\end{itemize}

The following folklore lemma characterises infinite sets of branches in the full binary tree.

\begin{lemma}\label{lemma:infinite}
Let $\Pi$ be a set of branches. Then $\Pi$ is infinite if and only if the set $W=\{w\mid\exists \pi\in\Pi \text{ s.t. } w \text{ belongs to } \pi\}$ contains a pseudo comb $(U,V)$, \ie $U\cup V\subseteq W$.
\end{lemma}

\begin{proof}
There exists an increasing sequence (for the prefix relation $\prefix$)  $(u_i)_{i\geq 0}$ of nodes such that from all $i\geq 0$ infinitely many branches in $\Pi$ go through $u_i$ and moreover at least one branch in $\Pi$ goes through $u_i\cdot 0$ (\resp through $u_i\cdot 1$). The existence of this sequence is by an immediate induction. 

Define $U$ as the set of prefixes of elements in the sequence $(u_i)_{i\geq 0}$: $U$ is a branch as the sequence $(u_i)_{i\geq 0}$ is increasing.

For all $i$, pick a branch $V_i$ that goes through $u_i$ but not from $u_{i+1}$ (it exists by definition of $u_i$). Then to obtain a pseudo comb $(U,V)$ such that $U\cup V\subseteq W$, it suffices to define $V=(\bigcup_{i\geq 0} V_i)\setminus U$.
\end{proof}

\begin{figure}
\begin{center}
\begin{tikzpicture}
\tikzset{>=stealth}
%%% path mode
  \node[draw,circle,inner sep=0mm,minimum size=9mm] (q) at (0,0) {\footnotesize$q,u$};  
  \node[draw,inner sep=1mm,minimum size=5mm,ellipse] (qp0) at (2,.8) {\footnotesize$q,u,q_0,q_1,0$};  
  \node[draw,inner sep=1mm,minimum size=5mm,ellipse] (qp1) at (2,-.8) {\footnotesize$q,u,q_0,q_1,1$};  
  \node[draw,circle,inner sep=0mm,minimum size=9mm] (p0) at (4.2,.8) {\footnotesize $q_{0},u0$};  
  \node[draw,circle,inner sep=0mm,minimum size=9mm] (p1) at (4.2,-.8) {\footnotesize $q_{1},u1$};  
  \node[draw,inner sep=1mm,minimum size=5mm] (p0A) at (4.2,2) {\footnotesize$q,u,q_0,q_1,0$};  
  \node[draw,inner sep=1mm,minimum size=5mm] (p1A) at (4.2,-2) {\footnotesize$q,u,q_0,q_1,1$};  
%%% check mode
  \node[draw,circle,inner sep=0mm,minimum size=9mm] (CMp1) at (7,2) {\footnotesize $q_{1},u1$};  
  \node[draw,circle,inner sep=0mm,minimum size=9mm] (CMp0) at (7,-2) {\footnotesize $q_{0},u0$};  
  \node[draw,circle,inner sep=0mm,minimum size=9mm] (CMq) at (9.5,0) {\footnotesize$q,u$};  
  \node[draw,inner sep=1mm,minimum size=5mm,ellipse] (CMp) at (7,0) {\footnotesize$q,u,q_0,q_1,0$};  
%%% Arcs
\draw[->] (q) to [in=180,out=90]  (qp0);\draw[->] (q) to [in=180,out=-90]  (qp1);
\draw[->] (qp0) -- (p0);\draw[->] (qp1) -- (p1);
\draw[->] (qp0) to [in=180,out=90] (p0A);\draw[->] (qp1) to [in=180,out=-90] (p1A);
\draw[->] (p0A) -- (p0);\draw[->] (p1A) -- (p1);
\draw[->] (CMp) --  (CMp0);\draw[->] (CMp) --  (CMp1);
\draw[->] (CMq) --  (CMp);
\draw[->] (p0A) -- (CMp1);\draw[->] (p1A) -- (CMp0);

\node (l) at (3,-3) {\footnotesize for any $(q,t(u),q_{0},q_{1})\in\Delta$};  
\node (path) at (3,3) {\footnotesize \emph{\underline{path mode}}};  
\node (check) at (8,3) {\footnotesize \emph{\underline{check mode}}};  
\node (l) at (8,-3) {\footnotesize for any $(q,t(u),q_{0},q_{1})\in\Delta$};  
\draw[dotted] (5.5,3.3) -- (5.5,-3.3);

\end{tikzpicture}
	\caption{Local structure of the arena of the acceptance game $\gameAccInf{\A,t}$.\label{fig:acceptance-game:AccInf}}
\end{center}
\end{figure}

Now, fix a tree $t$ and define an acceptance game for $\LAccInf{\A}$. 
There are two modes in the game (See Figure~\ref{fig:acceptance-game:AccInf} for the local structure of the arena): 
\emph{path} mode and \emph{check} mode and the game starts in \emph{path} mode. Hence, (\Eloise's) configurations in the game are elements of $Q\times\{0,1\}^*\times \{path,check\}$.
In \emph{path} mode, in a node $u$ with state $q$ \Eloise picks a transition $(q,t(u),q_0,q_1) \in \Delta$, and she chooses a direction $i\in \{0,1\}$.
Then \Eloise has two options. Either she moves down the pebble to $u\cdot i$ and updates the state to be $q_i$. Or she proposes \Abelard to change to \emph{check} mode: if he accepts, the pebble is moved down to $u\cdot (1-i)$ and the state is updated to $q_{(1-i)}$; if he refuses, the pebble is moved down to $u\cdot i$ and the state is updated to $q_i$
%\begin{itemize}
%\item She moves down the pebble to $u\cdot i$ and updates the state to be $q_i$.
%\item She proposes \Abelard to change to \emph{check} mode. If he accepts, the pebble is moved down to $u\cdot (1-i)$ and the state is updated to $q_{(1-i)}$; if he refuses, the pebble is moved down to $u\cdot i$ and the state is updated to $q_i$.
%\end{itemize}

In \emph{check} mode \Eloise plays alone: in a node $u$ with state $q$ she picks a transition $(q,t(u),q_0,q_1) \in \Delta$, and she chooses a direction $i\in \{0,1\}$; then the pebble is moved down to $u\cdot i$ and the state is updated to $q_i$. Note that there is no possible switch  from \emph{check} mode back to \emph{path} mode.

A play is won by \Eloise if one of the two following situations occurs. 
\begin{itemize}
	\item Eventually the players have switched to \emph{check} mode and the parity condition is satisfied.
	\item \Eloise proposed infinitely often \Abelard to switch the mode but he always refused. 
\end{itemize}
Call this game $\gameAccInf{\A,t}$. 

The next theorem states that it is an acceptance game for language $\LAccInf{\A}$.

\begin{theorem}\label{thm:LAccInf:acceptance}
One has $t\in \LAccInf{\A}$ if and only if \Eloise wins in $\gameAccInf{\A,t}$ from $(\qini,\epsilon,path)$.
\end{theorem}

\begin{proof}

Assume that \Eloise has a winning strategy $\phi$ in $\gameAccInf{\A,t}$ from configuration $(\qini,\epsilon,path)$.
With $\phi$ we associate a run $\run$ of $\A$ on $t$ and a pseudo comb $(U,V)$ as follows. We inductively associate with any node $u\in U$ a partial play $\play_u$ where \Eloise respects $\phi$ and that is always in \emph{path} mode; and we inductively associate with any node $v\in V$ a partial play $\play_v$ where \Eloise respects $\phi$ and where the mode has eventually been switched to \emph{check} mode. For this we let $\epsilon\in U$ and $\play_\epsilon = (\qini,\epsilon,path)$. Now assume that we defined $\play_u$ for some node $u\in U$ and let $(q,t(u),q_0,q_1)$ be the transition and let $i$ be the direction \Eloise plays from $\play_u$ when she respects $\phi$. Then we have two possible situations depending whether, right after playing $(q,t(u),q_0,q_1)$ and still respecting $\phi$, \Eloise proposes \Abelard to switch to \emph{check} mode. 
\begin{itemize}
\item If she does so we let $u\cdot (1-i)$ belongs to $V$ and we define $\play_{u\cdot (1-i)}$ as the partial play obtained by extending $\play_u$ by \Eloise choosing transition $(q,t(u),q_0,q_1)$ and direction $i$, followed by \Eloise proposing \Abelard to switch the mode and \Abelard accepting (hence, moving down the pebble in direction $(1-i)$ and attaching state $q_{(1-i)}$ to it). We let $u\cdot i$ belongs to $U$ and we define $\play_{u\cdot i}$ as the partial play obtained by extending $\play_u$ by \Eloise choosing transition $(q,t(u),q_0,q_1)$ and direction $i$, followed by \Eloise proposing \Abelard to switch the mode and \Abelard refusing (hence, moving down the pebble in direction $i$ and attaching state $q_{i}$ to it).
\item 
If \Eloise does not propose \Abelard to switch the mode we do not let $u\cdot (1-i)$ belongs to $V$. And we let $u\cdot i$ belongs to $U$ and we define $\play_{u\cdot i}$ as the partial play obtained by extending $\play_u$ by \Eloise choosing transition $(q,t(u),q_0,q_1)$ and direction $i$, followed by \Eloise not proposing \Abelard to switch the mode (hence, moving down the pebble in direction $i$ and attaching state $q_{i}$ to it).
\end{itemize}

The run $\run$ is defined by letting, for any $w\in U\cup V$, $\run(w)$ be the state attached to the pebble in the last configuration of $\play_w$. For those $w\notin U\cup V$ we define $\run(w)$ so that the resulting run is valid, which is always possible as we only consider complete automata. By construction, $\run$ is a run of $\A$ on $t$ and $(U,V)$ is a pseudo comb. Moreover with any branch $\pi$ that can be built as an initial sequence of nodes in $U$ followed by an infinite sequence of nodes in $V$ one can associate a play $\play_\pi$ in $\gameAccInf{\A,t}$ from $(\qini,\epsilon,path)$ where \Eloise respects $\phi$ (one simply considers the limit of the \emph{increasing} sequence of partial plays $\play_v$ where $v$ ranges those nodes nodes in $V$ along branch $\pi$). By construction $\pi$ is accepting as $\play_\pi$ fulfils the parity condition. Hence, by Lemma~\ref{lemma:infinite} we conclude that $\run$ contains infinitely many accepting branches, meaning that $t\in \LAccInf{\A}$.

Conversely, assume that \Eloise does not have a winning strategy in $\gameAccInf{\A,t}$ from $(\qini,\epsilon,path)$. By Borel determinacy, \Abelard has a winning strategy $\psi$ in $\gameAccInf{\A,t}$ from $(\qini,\epsilon,path)$. By contradiction, assume that there is a run $\run$ of $\A$ on $t$ that contains infinitely many accepting branches. By Lemma~\ref{lemma:infinite}, it follows that $\run$ contains a pseudo comb $(U,V)$ such that any branch that can be built as an initial sequence of nodes in $U$ followed by an infinite sequence of nodes in $V$ is an accepting branch. From $\run$ and $(U,V)$ we define a strategy $\phi$ of \Eloise in $\gameAccInf{\A,t}$ from $(\qini,\epsilon,path)$ as follows. Strategy $\phi$ uses as a memory either a node $u\in U$ if the play is in \emph{path} mode or a node $v\in V$ if the play is in \emph{check} mode; initially the memory is $u=\epsilon$. Now assume that the pebble is in some node $u\in U$ with state $q$ attached to it (one will inductively check that $\run(u)=q$). Then there are two possibilities.
\begin{itemize}
\item Both $u0$ and $u1$ belong to $U\cup V$: strategy $\phi$ indicates that \Eloise chooses transition $(q,t(u),\run(u0),\run(u1))$ and direction $i$ where $ui\in U$ and proposes \Abelard to switch to \emph{check} mode. Then the memory is updated to $u\cdot j$ where $j=i$ if the mode is unchanged and $j=1-i$ otherwise.
\item $u0$ (\resp $u1$) belong to $U$ but $u1$ (\resp $u0$) does not belong to $V$: strategy $\phi$ indicates that \Eloise chooses transition $(q,t(u),\run(u0),\run(u1))$ and chooses direction $0$ (\resp $1$) and does not propose \Abelard to switch to \emph{check} mode. Then the memory is updated to $u0$ (\resp $u1$).
\end{itemize}

Now consider a play $\play$ where \Eloise respects her strategy $\phi$ while \Abelard respects his strategy $\psi$. First, as $(U,V)$ is a pseudo comb and by definition of $\phi$, it follows that $\play$ only goes through nodes in $U\cup V$, and if $\play$ only goes through nodes in $U$ then \Eloise proposes infinitely often to \Abelard to switch to \emph{check} mode. Hence, as $\psi$ is winning for \Abelard one concludes that eventually the mode is switched in $\play$ and that the resulting play does not fulfil the parity condition. Now, it is easily seen that with $\play$ one associates a branch $\pi$ in the run $\run$ and that this branch can be built as an initial sequence of nodes in $U$ followed by an infinite sequence of nodes in $V$. As this branch is rejecting in $\run$ it leads a contradiction.
\end{proof}

One can modify $\gameAccInf{\A,t}$ so that to obtain an \emph{equivalent} game that has the form of a classical acceptance game. From this follows the fact that the languages of the form $\LAccInf{\A}$ are indeed $\omega$-regular. As the new game can be seen to be obtained from a Büchi automaton, this also permits to lower the acceptance condition.

\begin{theorem}\label{thm:LAccInf:reg}
Let $\A=\langle A, Q , q_{ini},\Delta,\col \rangle$ be a parity tree automaton using $d$ colours. Then there exists a \emph{Büchi} tree automaton $\A'=\langle A, Q' , q'_{ini},\Delta',\col' \rangle$ such that $\LAccInf{\A}=L(\A')$. Moreover $|Q'|=\mathcal{O}(d|Q|)$.
\end{theorem}

\begin{proof}
Start from game $\gameAccInf{\A,t}$ and observe that if one duplicates the control states (with a classical version and a starred version of each state) and add a Boolean flag \Eloise can indicate the direction she wants to follow and whether she proposes to switch to \emph{check} mode: \eg if she wants to choose transition $(q,t(u),q_0,q_1)$ and direction $1$ and not change the mode, she just moves to configuration $(q,u,q_0,q_1^*,\bot)$ in the new game; if she wants to choose transition $(q,t(u),q_0,q_1)$ and direction $0$ and offer \Abelard the option to switch to \emph{check} mode she moves to configuration $(q,u,q_0^*,q_1,\top)$. Now, we allow \Abelard to choose any direction but if the corresponding state is not starred then either one goes to a dummy winning configuration for \Eloise if the Boolean was $\bot$ and otherwise one changes the mode to \emph{check}. We indicate the \emph{check} mode in the control state and we use the same trick to let \Eloise impose the choice of the branch (if \Abelard does not follow her choice one ends up in the previous dummy configuration). It should be clear that the resulting game is an equivalent acceptance game when equipped with the following winning condition: \Eloise wins if either the dummy configuration is reached, or infinitely many configuration with Boolean $\top$ are visited but the play is always in \emph{path} mode or the play is eventually in \emph{check} mode and the parity condition holds. 
Now, the two first criteria are Büchi criteria while the third one is \emph{a priori} a parity condition. But as \Eloise plays alone in \emph{check} mode, she can indicate at some point that the smallest infinitely visited colour will be some (even) integer and that no other smallest colour will latter be visited: hence, if one stores the colour, go to a final state whenever it is visited and to a rejecting state if some smallest colour occurs, then one obtains a Büchi condition.
All together (combining the Büchi conditions in the usual way) one obtains an equivalent {Büchi} classical acceptance game, showing that $\LRejAtMostCount{\A}$ is $\omega$-regular and accepted by a Büchi automaton. 

The construction of $\A'$ is immediate from the final game and the size is linear in $d|Q|$ due to the fact that one needs to remember the smallest colour for the \emph{check} mode.
\end{proof}

\subsection{The Case of Languages $\LAccUnc{\A}$}

We now discuss the case of languages of the form $\LAccUnc{\A}$. For this we start with some key objects (\emph{accepting pseudo binary trees} and \emph{$k$-pseudo binary trees}) that are used to characterise runs with uncountably many accepting branches. Then, we describe two acceptance games: the first one is very simple while the second one is more involved but later permits to lower the acceptance condition to Büchi when showing that the  languages $\LAccUnc{\A}$ are accepted by tree automata with the classical semantics

\subsubsection{Accepting-Pseudo Binary Tree \& $k$-Pseudo Binary Tree}\label{section:PBT}

The key idea behind defining an acceptance game for $\LAccUnc{\A}$ for some tree $t$ is to exhibit an \emph{accepting pseudo binary tree} in a run of $\A$ over $t$. In a nutshell,
%\footnote{See Appendix for a formal definition and an illustration.}
 an accepting  pseudo binary tree is an infinite set $U$ of nodes with a tree-like structure between them and such that any branch that {has infinitely many prefixes in $U$} is accepting.

We now formally define \emph{accepting-pseudo binary trees} and \emph{$k$-pseudo binary trees} that characterise those runs that contains uncountably many accepting branches (Lemma~\ref{lemma:uncountable} and Lemma~\ref{lemma:uncountable_weak} below).

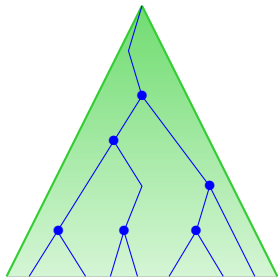
\begin{figure}[htb]
\begin{center}
\begin{tikzpicture}

\begin{scope}[scale = 0.6,]
\tikzset{>=stealth}

\fill[bottom color=LimeGreen!100!black!20,
top color=LimeGreen!100!black!70] (0,0) -- (-3,-6) -- (3,-6) -- cycle ;
\draw[LimeGreen,thick] (0,0) -- (-3,-6);
\draw[LimeGreen,thick] (0,0) -- (3,-6);
\draw[Blue] (0,0) -- (-0.3,-1)-- (0,-2);
\draw[Blue] (0,-2)-- (-2.5,-6);
\draw[Blue] (0,-2)-- (1.5,-4)-- (1.2,-5) -- (1.8,-6);
\draw[Blue] (1.5,-4) -- (2.5,-6);
\draw[Blue] (1.2,-5) -- (.6,-6);
\draw[Blue] (-.625,-3) -- (0,-4) -- (-.4,-5) -- (-.1,-6);
\draw[Blue] (-.4,-5) -- (-.7,-6);
\draw[Blue] (-1.875,-5) -- (-1.25,-6);
\node at (0,-2) {\textcolor{Blue}{$\bullet$}};
\node at (1.5,-4) {\textcolor{Blue}{$\bullet$}};
\node at (1.2,-5) {\textcolor{Blue}{$\bullet$}};
\node at (-.625,-3) {\textcolor{Blue}{$\bullet$}};
\node at (-.4,-5) {\textcolor{Blue}{$\bullet$}};
\node at (-1.85,-5) {\textcolor{Blue}{$\bullet$}};
\end{scope}

\end{tikzpicture}
\caption{An accepting-pseudo binary tree U: nodes in $U$ are marked by symbol $\textcolor{Blue}{\bullet}$ and all blue branches are accepting.}\label{sfigure:PBT}
\end{center}
\end{figure}

Let $\A=\langle A, Q , q_{ini},\Delta,\col \rangle$ be a parity tree automaton and let $\run$ be a run of $\A$ on some tree $t$.  

An \defin{accepting-pseudo binary tree} in $\run$ (see Figure~\ref{sfigure:PBT} for an illustration) is a subset $U\subseteq \{0,1\}^*$ of nodes such that
\begin{enumerate}[(i)] 
\item for all $u\in U$ there are $v,w\in U$ such that $v=u0v'$ and $w=u1w'$
for some $v'$ and $w' \in \{0,1\}^{*}$;
%
%$u\prefixstrict v$ and $u\prefixstrict w$ and neither $v\prefix w$ nor $w \prefix v$;
\item for all $v,w\in U$ the largest common prefix $u$ of $v$ and $w$ belongs to $U$;
\item any branch $\pi$ that goes through infinitely many nodes in $U$ is accepting.
\end{enumerate}

\begin{figure}[htb]
\begin{center}
\begin{tikzpicture}

\begin{scope}[scale = 0.6,]
\tikzset{>=stealth}

\fill[bottom color=LimeGreen!100!black!20,
top color=LimeGreen!100!black!70] (0,0) -- (-3,-6) -- (3,-6) -- cycle ;
\draw[LimeGreen,thick] (0,0) -- (-3,-6);
\draw[LimeGreen,thick] (0,0) -- (3,-6);
\draw[Grey,dotted, thick] (0,0) -- (-0.3,-1)-- (0,-2);
\draw[Grey,dashed, thick] (0,-2)-- (-2.5,-6);
\draw[Grey,dashed, thick] (0,-2)-- (1.5,-4)-- (1.2,-5) -- (1.8,-6);
\draw[Grey,dashed, thick] (1.5,-4) -- (2.5,-6);
\draw[Grey,dashed, thick] (1.2,-5) -- (.6,-6);
\draw[Grey,dashed, thick] (-.625,-3) -- (0,-4) -- (-.4,-5) -- (-.1,-6);
\draw[Grey,dashed, thick] (-.4,-5) -- (-.7,-6);
\draw[Grey,dashed, thick] (-1.875,-5) -- (-1.25,-6);
\node at (0,-2) {\textcolor{Blue}{$\bullet$}};
\node at (1.5,-4) {\textcolor{Blue}{$\bullet$}};
\node at (1.2,-5) {\textcolor{Blue}{$\bullet$}};
\node at (-.625,-3) {\textcolor{Blue}{$\bullet$}};
\node at (-.4,-5) {\textcolor{Blue}{$\bullet$}};
\node at (-1.85,-5) {\textcolor{Blue}{$\bullet$}};
\node(min) at (-3,-2) {$\min = k$};
\draw[->] (min) to [out=-90,in=135] (-0.4,-2.5);
\draw[->] (min) to [out=90,in=35] (0.8,-3);
\draw[->] (min) to [out=-110,in=155] (-1.35,-4);
\draw[->] (min) to [out=-100,in=150] (-0.1,-4);
\draw[->] (min) .. controls (-4,2) and  (6,-2) .. (1.35,-4.7);
\node (shift) at (-6,0) {};
\end{scope}

\end{tikzpicture}
\caption{A $k$-pseudo binary tree $U$: nodes in $U$ are marked by symbol $\textcolor{Blue}{\bullet}$.}\label{sfigure:kPBT}
\end{center}
\end{figure}
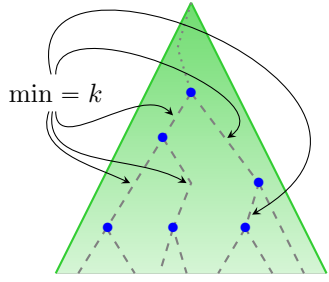

We now give a stronger notion than accepting-pseudo binary tree. For this, let $k$ be some even colour. A \defin{$k$-pseudo binary tree} in $\run$ (see Figure~\ref{sfigure:kPBT} for an illustration) is a subset $U\subseteq \{0,1\}^*$ of nodes such that
\begin{enumerate}[(i)] 
\item for all $u\in U$ there are $v,w\in U$ such that $v=u0v'$ and $w=u1w'$
for some $v'$ and $w' \in \{0,1\}^{*}$;
%\item for all $u\in U$ there are $v,w\in U$ such that $u\prefixstrict v$ and $u\prefixstrict w$ and neither $v\prefix w$ nor $w \prefix v$;
\item for all $v,w\in U$ the largest common prefix $u$ of $v$ and $w$ belongs to $U$;
\item for all $u,v\in U$ such that $u\prefixstrict v$, one has $\min\{\col(\run(w))\mid u\prefix w\prefix v\} = k$.
\end{enumerate}

The following lemma characterises runs that contains an uncountable sets of accepting branches. Its proof is a direct consequence of \cite[Lemma~2]{BN95}

\begin{lemma}\label{lemma:uncountable}
Let $\run$ be a run. Then $\run$ contains uncountably many accepting branches if and only if it contains a $k$-pseudo binary tree for some even colour $k$.
\end{lemma}

As any $k$-pseudo binary tree is an accepting-pseudo binary tree, we directly have the following Lemma from  Lemma~\ref{lemma:uncountable}.

\begin{lemma}\label{lemma:uncountable_weak}
Let $\run$ be a run. Then $\run$ contains uncountably many accepting branches if and only if it contains an accepting-pseudo binary tree.
\end{lemma}

\subsubsection{The Acceptance Game $\gameAccUnc{\A,t}$}

\begin{figure}[htb]
\begin{center}
\begin{tikzpicture}
\tikzset{>=stealth}
  \node[draw,circle,inner sep=0mm,minimum size=9mm] (q) at (-2.5,0) {\footnotesize$q,u$};  
  \node[draw,inner sep=1mm,minimum size=5mm,ellipse] (p) at (0,0) {\footnotesize$q,u,q_0,q_1$};  
  \node[draw,circle,inner sep=0mm,minimum size=9mm] (p0) at (2,0.8) {\footnotesize $q_{0},u0$};  
  \node[draw,circle,inner sep=0mm,minimum size=9mm] (p1) at (2,-0.8) {\footnotesize $q_{1},u1$};  
  \node[draw,inner sep=1mm,minimum size=5mm] (pE) at (4,0) {\footnotesize$q,u,q_0,q_1$};  

\draw[->] (q) --  (p);
\draw[->] (p) to [in=180,out=70] (p0);\draw[->] (p) to [in=180,out=-70] (p1);
\draw[->] (pE) to [in=0,out=110] (p0);\draw[->] (pE) to [in=0,out=-110] (p1);
\draw[->] (p) -- (pE);
\node (l) at (-1.5,-1) {\footnotesize for any $(q,t(u),q_{0},q_{1})\in\Delta$};  
\end{tikzpicture}
	\caption{Local structure of the arena of the acceptance game $\gameAccUnc{\A,t}$.\label{fig:acceptance-game:AccUnc}}
\end{center}
\end{figure}

Fix a tree $t$ and define an acceptance game for $\LAccUnc{\A}$. In this game (see Figure~\ref{fig:acceptance-game:AccUnc} for the local structure of the arena) 
 the two players move a pebble along a branch of $t$ in a top-down manner: to the pebble is attached a state, and the colour of the state gives the colour of the configuration. Hence, (\Eloise's main) configurations in the game are elements of $Q\times\{0,1\}^*$. In a node $u$ with state $q$ \Eloise picks a transition $(q,t(u),q_0,q_1) \in \Delta$, and then
\Eloise has two possible options. Either she chooses a direction $0$ or $1$ or she lets \Abelard choose a direction $0$ or $1$.
Once the direction $i\in\{0,1\}$ is chosen, the pebble is moved down to $u\cdot i$ and the state is updated to $q_i$. 
A play is won by \Eloise if and only if \begin{enumerate}[(1)] \item the parity condition is satisfied and \item \Eloise lets \Abelard infinitely often choose the direction during the play. \end{enumerate} Call this game $\gameAccUnc{\A,t}$. 

The next theorem states that it is an acceptance game for language $\LAccUnc{\A}$.

\begin{theorem}\label{thm:LAccUnc:acceptance}
\Eloise wins in $\gameAccUnc{\A,t}$ from $(\qini,\epsilon)$ if and only if $t\in \LAccUnc{\A}$.
\end{theorem}

\begin{proof}
In the following proof for a set $X\subseteq \{0,1\}^*$ we denote by $Pref(X)$ the set of prefixes of elements in $X$, \ie $\mathrm{Pref}(X)=\{u\mid \exists v\in X\text{ s.t. }u\prefix v\}$.

Assume that \Eloise has a winning strategy $\phi$ in $\gameAccUnc{\A,t}$ from $(\qini,\epsilon)$.
With $\phi$ we associate a run $\run$ of $\A$ on $t$ and an accepting-pseudo binary tree $U$ as follows. We inductively define $U$ and $\mathrm{Pref}(U)$ and associate with any node $u\in \mathrm{Pref}(U)$ a partial play $\play_u$ where \Eloise respects $\phi$. For this we let $\epsilon\in \mathrm{Pref}(U)$ and we set $\play_\epsilon = (\qini,\epsilon)$. 

Now assume that we have defined $\play_u$ for some node $u\in \mathrm{Pref}(U)$. Then let $(q,t(u),q_0,q_1)$ be the transition \Eloise plays from $\play_u$ when she respects $\phi$. Then we have two possible situations depending whether, right after playing $(q,t(u),q_0,q_1)$ and still respecting $\phi$, \Eloise chooses the direction or let \Abelard make that choice. If she chooses the direction, let $i$ be this direction: then one lets $ui\in \mathrm{Pref}(U)$ and defines $\play_{u\cdot i}$ as the partial play obtained by extending $\play_u$ by \Eloise choosing transition $(q,t(u),q_0,q_1)$, followed by \Eloise choosing direction $i$. If she lets \Abelard choose the direction, one lets $u$ belongs to $U$ and lets both $u0$ and $u1$ belongs to $\mathrm{Pref}(U)$ and defines $\play_{u\cdot i}$ for $i\in\{0,1\}$ as the partial play obtained by extending $\play_u$ by \Eloise choosing transition $(q,t(u),q_0,q_1)$, followed by \Eloise letting \Abelard choose the direction and \Abelard picking direction $i$. Note that for any $ui\in \mathrm{Pref}(U)$, $\play_{u\cdot i}$ ends with the pebble on $u\cdot i$ with state $q_i$ attached to it, equivalently in configuration $(q_i,ui)$.

The run $\run$ is defined by letting, for any $u\in \mathrm{Pref}(U)$, $\run(u)$ be the state attached to the pebble in the last configuration of $\play_u$. For those $u\notin \mathrm{Pref}(U)$ we define $\run(u)$ so that the resulting run is valid, which is always possible as we only consider complete automata. By construction, $\run$ is a run of $\A$ on $t$. Moreover, with any branch $\pi$ consisting only of nodes in $\mathrm{Pref}(U)$, one can associate a play $\play_\pi$ in $\gameAccUnc{\A,t}$ from $(\qini,\epsilon)$ where \Eloise respects $\phi$ (one simply considers the limit of the \emph{increasing} sequence of partial plays $\play_u$ where $u$ ranges over nodes along branch $\pi$). As $\play_\pi$ is winning it follows easily that $U$ is a pseudo binary tree (indeed, condition $(i)$ and $(ii)$ from the definition of an accepting-pseudo binary tree are immediate, while condition $(iii)$ follows from the fact that $\play_\pi$ is winning). Hence, from Lemma~\ref{lemma:uncountable_weak} we conclude that $\run$ contains uncountably many accepting branches, meaning that $t\in \LAccUnc{\A}$.

Conversely, assume that there is a run $\run$ of $\A$ on $t$ that contains uncountably many accepting branches. By Lemma~\ref{lemma:uncountable_weak}, it follows that $\run$ contains an accepting-pseudo binary tree $U$.

From $\run$ and $U$ we define a strategy $\phi$ of \Eloise in $\gameAccUnc{\A,t}$ from $(\qini,\epsilon)$ as follows. Strategy $\phi$ uses as a memory a node $v\in \mathrm{Pref}(U)$, and initially $v=\epsilon$. Now assume that the pebble is on some node $v$ with state $q$ attached to it (one will inductively check that $v\in \mathrm{Pref}(U)$ and that $\run(v)=q$). Then we have two possibilities.
\begin{itemize}
\item Assume $v\in U$. Both $v0$ and $v1$ belong to $\mathrm{Pref}(U)$: strategy $\phi$ indicates that \Eloise chooses transition $(q,t(v),\run(v0),\run(v1))$ and let \Abelard choose the direction, say $i$. Then the memory is updated to $v\cdot i$.
\item Assume $v\notin U$. Hence, $v\cdot i$ belong to $\mathrm{Pref}(U)$ for only one $i\in\{0,1\}$: strategy $\phi$ indicates that \Eloise chooses transition $(q,t(v),\run(v0),\run(v1))$ and chooses direction $i$. Then the memory is updated to $v\cdot i$.
\end{itemize}

Now consider a play $\play$ where \Eloise respects her strategy $\phi$. It is easily seen that with $\play$ one associates a branch $\pi$ in the run $\run$ and that this branch goes only through nodes in $\mathrm{Pref}(U)$. From this observation and from the definition of an accepting-pseudo binary tree, we conclude that $\play$ is winning for \Eloise (it satisfies the parity condition as $\pi$ does and in $\play$ \Eloise lets \Abelard choose the direction infinitely often, namely whenever her memory $v$ belongs to $V$). Hence, we conclude that strategy $\phi$ is winning from $(\qini,\epsilon)$.
\end{proof}

\subsubsection{The Acceptance Game $\gameAccUnck{\A,t}$}

One can modify $\gameAccUnc{\A,t}$ to obtain an \emph{equivalent} game that has the form of a classical acceptance game. From this follows the fact that the languages of the form $\LAccUnc{\A}$ are indeed $\omega$-regular. 
Nevertheless, using a more involved game than $\gameAccUnc{\A,t}$ one can obtain a stronger result where the acceptance condition is lowered to a Büchi condition. We now describe this game.

\begin{figure}[htb]
\begin{center}
\begin{tikzpicture}
\tikzset{>=stealth}

%%% wait mode
\begin{scope}
  \node[draw,circle,inner sep=0mm,minimum size=9mm] (CMp1) at (7,-2) {\footnotesize $q_{1},u1$};  
  \node[draw,circle,inner sep=0mm,minimum size=9mm] (CMp0) at (7,2) {\footnotesize $q_{0},u0$};  
  \node[draw,circle,inner sep=0mm,minimum size=9mm] (CMq) at (4.5,0) {\footnotesize$q,u$};  
  \node[draw,inner sep=1mm,minimum size=5mm,ellipse] (CMp) at (7,0) {\footnotesize$q,u,q_0,q_1,0$};  
%%% Arcs
\draw[->] (CMp) --  (CMp0);\draw[->] (CMp) --  (CMp1);
\draw[->] (CMq) --  (CMp);
%%%Mode info
\node (check) at (6,3) {\footnotesize \emph{\underline{wait mode}}};  
\node (l) at (6,-3) {\footnotesize for any $(q,t(u),q_{0},q_{1})\in\Delta$};  
\node (l) at (6,-3.5) {\footnotesize for any even colour $k$};  
\draw[dotted] (8.5,3.3) -- (8.5,-3.8);
\end{scope}

%%%check mode
\begin{scope}[xshift = 12cm]
  \node[draw,circle,inner sep=0mm,minimum size=9mm] (q) at (-2.5,0) {\footnotesize$q,u$};  
  \node[draw,inner sep=1mm,minimum size=5mm,ellipse] (p) at (0,0) {\footnotesize$q,u,q_0,q_1$};  
  \node[draw,circle,inner sep=0mm,minimum size=9mm] (p0) at (2,0.8) {\footnotesize $q_{0},u0$};  
  \node[draw,circle,inner sep=0mm,minimum size=9mm] (p1) at (2,-0.8) {\footnotesize $q_{1},u1$};  
  \node[draw,inner sep=1mm,minimum size=5mm] (pE) at (4,0) {\footnotesize$q,u,q_0,q_1$};  

\draw[->] (q) --  (p);
\draw[->] (p) to [in=180,out=70] (p0);\draw[->] (p) to [in=180,out=-70] (p1);
\draw[->] (pE) to [in=0,out=110] (p0);\draw[->] (pE) to [in=0,out=-110] (p1);
\draw[->] (p) -- (pE);
\node (l) at (0,-3) {\footnotesize for any $(q,t(u),q_{0},q_{1})\in\Delta$};  
\node (check) at (0,3) {\footnotesize \emph{\underline{$check^k$ mode}}};  
\end{scope}

\draw[->] (CMp) to [out=20,in=160] (p0);
\draw[->] (CMp) to [out=-20,in=200] (p1);

\end{tikzpicture}
	\caption{Local structure of the arena of the acceptance game $\gameAccUnck{\A,t}$.\label{fig:acceptance-game:AccUnck}}
\end{center}
\end{figure}

Fix a tree $t$ and define an acceptance game for $\LAccUnc{\A}$. There are two modes in the game (See Figure~\ref{fig:acceptance-game:AccUnck} for the local structure of the arena): \emph{wait} mode and \emph{check} mode and the game starts in \emph{wait} mode. Moreover the \emph{check} mode is parametrised by a colour $k$. Again, the two players move a pebble along a branch of $t$ in a top-down manner. Hence, (main) configurations in the game are elements of $Q\times\{0,1\}^*\times \{wait,check^0,\dots,check^{2\ell}\}$ where $\{0,\dots,2\ell\}$ are the even colours used by $\A$. In \emph{wait} mode \Eloise plays alone: in a node $u$ with state $q$ she picks a transition $(q,t(u),q_0,q_1) \in \Delta$, and she chooses a direction $i\in \{0,1\}$; then the pebble is moved down to $u\cdot i$ and the state is updated to $q_i$. When moving the pebble down she can decide to switch the mode to some $check^k$ (for any \emph{even} colour $k$). Once entered $check^k$ mode the play stays in that mode forever and goes as follows.
In a node $u$ with state $q$ \Eloise picks a transition $(q,t(u),q_0,q_1) \in \Delta$, and then
she has two possible options. Either she chooses a direction $0$ or $1$ or she lets \Abelard choose a direction $0$ or $1$.
Once the direction $i\in\{0,1\}$ is chosen, the pebble is moved down to $u\cdot i$ and the state is updated to $q_i$. 
A play is won by \Eloise if and only if \begin{enumerate}[(1)] \item it eventually enters some $check^k$ mode and \item it goes infinitely often through configurations in $\{(q,u,check^k)\mid \col(q)=k\}$, \item it never visits a configuration in $\{(q,u,check^k)\mid \col(q)<k\}$, \item \Eloise lets \Abelard infinitely often choose the direction during the play, and between two such situations the smallest colour visited is always $k$.
\end{enumerate} Call this game $\gameAccUnck{\A,t}$. 

The next theorem states that it is an acceptance game for language $\LAccUnc{\A}$. Note that its proof is a refinement of the one of Theorem~\ref{thm:LAccUnc:acceptance}.

\begin{theorem}\label{thm:LAccUnck:acceptance}
One has $t\in \LAccUnc{\A}$ if and only if \Eloise wins in $\gameAccUnck{\A,t}$ from $(\qini,\epsilon,wait)$.
\end{theorem}

\begin{proof}
In the following proof for a set $X\subseteq \{0,1\}^*$ we denote by $Pref(X)$ the set of prefixes of elements in $X$, \ie $\mathrm{Pref}(X)=\{u\mid \exists v\in X\text{ s.t. }u\prefix v\}$.

Assume that \Eloise has a winning strategy $\phi$ in $\gameAccUnck{\A,t}$ from $(\qini,\epsilon,wait)$.
With $\phi$ we associate a run $\run$ of $\A$ on $t$ and a $k$-pseudo binary tree $U$ (for some $k$ to be defined later) as follows. We inductively define $U$ and $\mathrm{Pref}(U)$ and associate with any node $u\in \mathrm{Pref}(U)$ a partial play $\play_u$ where \Eloise respects $\phi$. For this we let $\epsilon\in \mathrm{Pref}(U)$ and we let $\play_\epsilon = (\qini,\epsilon,wait)$. 

Now assume that we have defined $\play_u$ for some node $u\in \mathrm{Pref}(U)$ and that the mode in $\play_u$ is always $wait$. Then let $(q,t(u),q_0,q_1)$ be the transition and let $i$ be the direction \Eloise plays from $\play_u$ when she respects $\phi$. If she does not change the mode, then one lets $ui\in \mathrm{Pref}(U)$ and defines $\play_{u\cdot i}$ as the partial play obtained by extending $\play_u$ by \Eloise choosing transition $(q,t(u),q_0,q_1)$, followed by \Eloise choosing direction $i$ and keeping the mode to $wait$. If she changes the modes to $check^k$, then one lets $ui\in \mathrm{Pref}(U)$ and defines $\play_{u\cdot i}$ as the partial play obtained by extending $\play_u$ by \Eloise choosing transition $(q,t(u),q_0,q_1)$, followed by \Eloise choosing direction $i$ and changing the mode to $check^k$ (this $k$ is the one such that $U$ is a $k$-pseudo binary tree $U$). 

Now assume that we {have} defined $\play_u$ for some node $u\in \mathrm{Pref}(U)$ and that the mode in $\play_u$ has been switched from $wait$ to $check^k$. Then let $(q,t(u),q_0,q_1)$ be the transition \Eloise plays from $\play_u$ when she respects $\phi$. Then we have two possible situations depending whether, right after playing $(q,t(u),q_0,q_1)$ and still respecting $\phi$, \Eloise chooses the direction or lets \Abelard make that choice. If she chooses the direction, let $i$ be this direction: then one lets $ui\in \mathrm{Pref}(U)$ and defines $\play_{u\cdot i}$ as the partial play obtained by extending $\play_u$ by \Eloise choosing transition $(q,t(u),q_0,q_1)$, followed by \Eloise choosing direction $i$. If she lets \Abelard choose the direction, one lets $u$ belongs to $U$ and lets both $u0$ and $u1$ belongs to $\mathrm{Pref}(U)$ and defines $\play_{u\cdot i}$ for $i\in\{0,1\}$ as the partial play obtained by extending $\play_u$ by \Eloise choosing transition $(q,t(u),q_0,q_1)$, followed by \Eloise letting \Abelard choose the direction and \Abelard picking direction $i$. Note that for any $ui\in \mathrm{Pref}(U)$, $\play_{u\cdot i}$ ends with the pebble on $u\cdot i$ with state $q_i$ attached to it, equivalently in configuration $(q_i,ui)$.

The run $\run$ is defined by letting, for any $u\in \mathrm{Pref}(U)$, $\run(u)$ be the state attached to the pebble in the last configuration of $\play_u$. For those $u\notin \mathrm{Pref}(U)$ we define $\run(u)$ so that the resulting run is valid, which is always possible as we only consider complete automata. By construction, $\run$ is a run of $\A$ on $t$. Moreover with any branch $\pi$ consisting only of nodes in $\mathrm{Pref}(U)$ one can associate a play $\play_\pi$ in $\gameAccUnck{\A,t}$ from $(\qini,\epsilon,wait)$ where \Eloise respects $\phi$ (one simply considers the limit of the \emph{increasing} sequence of partial plays $\play_u$ where $u$ ranges over nodes along branch $\pi$). As $\play_\pi$ is winning it follows easily that $U$ is a $k$-pseudo binary tree (indeed, condition $(i)$ and $(ii)$ from the definition of a $k$-pseudo binary tree are immediate, while condition $(iii)$ follows from the definition of the winning condition and of the fact that $\play_\pi$ is winning). Moreover $\pi$ is accepting as the smallest colour infinitely often visited is $k$. As there are uncountably many branches $\pi$ consisting only of nodes in $\mathrm{Pref}(U)$ we conclude that $\run$ contains uncountably many accepting branches, meaning that $t\in \LAccUnc{\A}$.

Conversely, assume that there is a run $\run$ of $\A$ on $t$ that contains uncountably many accepting branches. By Lemma~\ref{lemma:uncountable}, it follows that $\run$ contains a $k$-pseudo binary tree $U$. Let $X = \{x\in \mathrm{Pref}(U)\mid \run(x)>k\}$: then by definition of a $k$-pseudo binary tree we conclude that $X$ is finite and has a minimal element for the prefix relation (with the convention that if $X$ is empty this minimum is set to be the root $\epsilon$); call $r$ this minimum. Note that there is also a minimum element $u_0$ in $U$ (for the prefix relation) and that $r\prefixstrict u_0$.

From $\run$ and $U$ we define a strategy $\phi$ of \Eloise in $\gameAccUnck{\A,t}$ from $(\qini,\epsilon,wait)$ as follows. Strategy $\phi$ uses as a memory a node $v\in \mathrm{Pref}(U)$ and initially $v=\epsilon$; moreover as long as $v\prefix r$ the play will be in $wait$ mode. Now assume that the pebble is on some node $v$ with state $q$ attached to it (one will inductively check that $v\in \mathrm{Pref}(U)$ and that $\run(v)=q$). Then we have several possibilities.
\begin{itemize}
\item The mode is $wait$ (\ie $v\prefix r\prefixstrict u_0$): strategy $\phi$ indicates that \Eloise chooses transition $(q,t(v),\run(v0),\run(v1))$, goes to direction $i$ where $i$ is such that $vi\prefix u_0$, and stay in mode $wait$ except if $v=r$ where the mode is switched to $check^k$.
\item The mode is $check^k$ and $v\in U$. Both $v0$ and $v1$ belong to $\mathrm{Pref}(U)$: strategy $\phi$ indicates that \Eloise chooses transition $(q,t(v),\run(v0),\run(v1))$ and let \Abelard choose the direction, say $i$. Then the memory is updated to $v\cdot i$.
\item The mode is $check^k$ and $v\notin U$. Hence, $v\cdot i$ belong to $\mathrm{Pref}(U)$ for only one $i\in\{0,1\}$: strategy $\phi$ indicates that \Eloise chooses transition $(q,t(v),\run(v0),\run(v1))$ and chooses direction $i$. Then the memory is updated to $v\cdot i$.
\end{itemize}

Now consider a play $\play$ where \Eloise respects her strategy $\phi$. It is easily seen that with $\play$ one associates a branch $\pi$ in the run $\run$ and that this branch goes only through nodes in $\mathrm{Pref}(U)$. From this observation and from the definition of a $k$-pseudo binary tree, we conclude that $\play$ is winning for \Eloise hence, that strategy $\phi$ is winning from $(\qini,\epsilon,wait)$.
\end{proof}

\subsubsection{Languages of the Form $\LAccUnc{\A}$ Are Büchi Regular}
Thanks to Theorem~\ref{thm:LAccUnck:acceptance} we can easily prove that any language of the form $\LAccUnc{\A}$ can be accepted by a \emph{Büchi} automaton.

\begin{theorem}\label{thm:LAccUnc:reg}
Let $\A=\langle A, Q , q_{ini},\Delta,\col \rangle$ be a parity tree automaton using $d$ colours. Then there exists a \emph{Büchi} tree automaton $\A'=\langle A, Q' , q'_{ini},\Delta',\col' \rangle$ such that $\LAccUnc{\A}=L(\A')$. Moreover $|Q'|=\mathcal{O}(d|Q|)$.
\end{theorem}

\begin{proof}
One can easily transform game $\gameAccUnck{\A,t}$ to obtain an equivalent game that is the acceptance game of some tree automaton $\A'$ with the classical semantics. The construction is very similar to the one we had for the other cases and we omit the details here. It simply suffices to notice that the winning condition in $\gameAccUnck{\A,t}$ is a conjunction of Büchi conditions, hence can be rephrased as a Büchi condition (up to adding some flags).
\end{proof}

%%%%%%%%%%%% TOPO %%%%%%%%%%%%
\section{Checking Topological Largeness of Accepting Branches}\label{sec:results-topo}

We now consider the case of automata with topological bigness constraints and we prove that languages of the form $\LLarge{\A}$ are always $\omega$-regular (Theorem~\ref{thm:LLarge:reg}). This acceptance condition is referred to as the best model of a \emph{fair adversary} in~\cite{VVK05}, and finite games where \Eloise plays against such an adversary has been studied and solved in~\cite{ACV10}. We first characterise large set of branches (Lemma~\ref{lemma:largeVSdense}), then based on this, we define an acceptance game for $\LLarge{\A}$ and finally we transform it so that to obtain an \emph{equivalent} game that has the form of a classical acceptance game from which one extracts $\A'$.

\os{Voir si on peut pas rentrer tout ça dans la preuve du lemme}
Banach-Mazur theorem gives a game characterization of large and meager sets of branches (see for instance \cite{Kechris,Graedel08}). The \defin{Banach-Mazur game} on a tree $t$,  is a two-player game where \Abelard and \Eloise choose alternatively a node in the tree, forming a branch: \Abelard chooses first a node and then \Eloise chooses a descendant of the previous node and \Abelard chooses a descendant of the previous node and so on forever. In this game it is \textbf{always \Abelard that starts} a play. 

Formally a \defin{play} is an infinite sequence $u_1,u_2,\ldots$ of words in $\{0,1\}^+$, and the branch associated with this play is $u_1u_2\cdots$. A \defin{strategy} for \Eloise is a mapping $\phi : (\{0,1\}^+)^+ \rightarrow \{0,1\}^+$ that takes as input a finite sequence of words, and outputs a word. A play $u_1,u_2,\ldots$ \defin{respects} $\phi$ if for all $i\geq 1$, $u_{2i}= \phi(u_1,\ldots,u_{2i-1})$. We define $Outcomes(\phi)$ as the set of plays that respect $\phi$ and $\mathcal{B}(\phi)$ as the set branches associated with the plays in $Outcomes(\phi)$. 

The Banach-Mazur theorem {(see \footnote{In \cite{Graedel08} the players of the Banach-Mazur game are called $0$ and $1$ and Player $0$ corresponds to \Abelard while player $1$ corresponds to \Eloise. Hence, when using a statement from \cite{Graedel08} for our setting one has to keep this in mind as well as the fact that one must replace the winning condition by its complement (hence, replacing “meager” by “large”).}
 \eg \cite[Theorem~4]{Graedel08})} states that a set of branches $B$ is large if and only if there exists a strategy $\phi$ for \Eloise such that $\mathcal{B}(\phi)\subseteq B$. % {Hence, if one thinks of $B$ as a winning condition for \Eloise (\ie she wins a play if and only if it belongs to $B$), it means that those sets $B$ for which she has a winning strategy are exactly the large ones.}

Furthermore a folk result {(see \eg \cite[Theorem~9]{Graedel08})} about Banach-Mazur games states that {when $B$ is Borel}\footnote{{This statement holds as soon as the Banach-Mazur games are determined and hence, in particular for Borel sets.}} one can look only at “simple” strategies, defined as follows. A \defin{decompo\-si\-tion-invariant strategy} is a mapping $f:\{0,1\}^* \rightarrow \{0,1\}^+$ and we associate with $f$ the strategy $\phi_f$ defined by $\phi_f(u_1,\ldots,u_{k})=f(u_1\cdots u_k)$. Finally, we define $Outcomes(f) = Outcomes(\phi_f)$ and $\mathcal{B}(f)=\mathcal{B}(\phi_f)$. The folk result states that for any {Borel} set of branches $B$,  there exists a strategy $\phi$ such that $Outcomes(\phi)\subseteq B$  if and only if there exists a decomposition-invariant strategy $f$ such that $\mathcal{B}(f)\subseteq B$.

Call a set of nodes $W\subseteq \{0,1\}^*$ \defin{dense} if  $\forall u\in \{0,1\}^*$, $\exists w\in W$ such that $u\prefix w$. Given a dense set of nodes $W$, the set of \defin{branches supported by $W$}, $\mathcal{B}(W)$ is the set of branches $\pi$ that have infinitely many prefixes in $W$. Using the existence of decomposition-invariant winning strategies in Banach-Mazur games, the following lemma characterises large sets of branches.

\begin{lemma}\label{lemma:largeVSdense}
A set of branches $B\subseteq \{0,1\}^\omega$ is large if and only if there exists a dense set of nodes  $W\subseteq \{0,1\}^*$ such that $\mathcal{B}(W) \subseteq B$.
\end{lemma}

\begin{proof}
Assume that $B$ is large and let $f$ be a decomposition-invariant strategy for \Eloise in the associated Banach-Mazur game. Consider the set:
\[
 W = \{ vf(v) \mid v \in \{0,1\}^* \}.
\]
The set $W$ is dense (as for all $v \in t$, $v \prefixstrict vf(v) \in W$). We claim that $\mathcal{B}(W)$ is included in $B$. Let $\pi$ be a branch 
in $\mathcal{B}(W)$. As $\pi$ has infinitely many prefixes in $W$, there exists
a sequence of words $u_{1},u_{2},\cdots$ such that $u_{1}f(u_{1}) \prefixstrict u_{2}f(u_{2}) \prefixstrict \cdots \prefixstrict \pi$. As the lengths of the $u_{i}$ are strictly increasing, there exists a sub-sequence $(v_i)_{i\geq 1}$ of $(u_i)_{i\geq 1}$ such that for all
$i \geq 1$, $v_{i}f(v_{i}) \prefixstrict v_{i+1}$. Now, consider the play in the Banach-Mazur game where \Abelard first move to $v_{1}$ and then \Eloise responds by going to $v_{1}f(v_{1})$. Then \Abelard moves to $v_{2}$ (which is possible as $v_{1}f(v_{1}) \prefixstrict v_{2}$) and \Eloise moves to $v_{2}f(v_{2})$. An so on. In this play \Eloise respects the strategy $f$ and therefore wins. Hence, the branch $\pi$ associated to this play belongs to $B$.

Conversely let $W$ be a dense set of nodes such that $\mathcal{B}(W) \subseteq B$. To show that $B$ is large, we define a decomposition-invariant strategy $f$ for \Eloise in the associated Bannach-Mazur game.
 For all node $u$ we pick  $v$ of $W$ such that $u$ is a strict prefix of $v$ (since $W$ is dense there must always exist such a $v$). Let $v=uu'$ and fix $f(u)=u'$. A play where \Eloise respects $f$ goes through infinitely many 
nodes in $W$ (as $f$ always points to an element in $W$). Hence, the branch associated with the play belongs to $\mathcal{B}(W)\subseteq B$ which shows that $f$ is winning for \Eloise.\end{proof}

Fix a tree $t$ and define
%\footnote{See formal definitions and details in Appendix.} 
an acceptance game $\gameAccLarge{\A,t}$ for $\LLarge{\A}$. In this game, \Eloise describes a run $\run$ together with a dense set $U$ of nodes while \Abelard tries either to prove that $U$ is not dense or that there is a rejecting branch in $\mathcal{B}(U)$. The way \Eloise describes a run is as usual (she proposes valid transitions); the way she describes $U$ is by \begin{inparaenum}[(1)] \item indicating explicitly when a node is in $U$ and; \item at each node giving a direction $i$ that  should {lead (by iteratively following the directions) to a node in $U$}. 
%go to look for the next node in $U$.
\end{inparaenum} \Abelard chooses the direction: if it does not select $i$ and does not go to a node in $U$ the colour is a large even one (preventing him not to follow \Eloise forever); if he chooses $i$ but does not go to a node in $U$ the colour is a large odd one (forcing \Eloise to describe a dense set $U$); and if he chooses $i$ and goes to a node in $U$ the colour is the smallest one seen since the last visit to a node in $U$ (and it is computed in the game).

Before formally constructing the game we need the following lemma. A \defin{direction mapping} is a mapping $d:\{0,1\}^* \rightarrow \{0,1\}$, and given a set of nodes $U$, we say that $d$ \defin{points to $U$} if for every node $v$ there exists $i_1,\ldots,i_k\in\{0,1\}$ such that $v i_1\cdots i_k \in U$ and for all $1\leq j\leq k$, $i_j = d(v i_1 \cdots i_{j-1})$.

\begin{lemma}\label{directions}
A  set of nodes $U$ is dense if and only if there exists a direction mapping that points to $U$.
\end{lemma}
\begin{proof}
Assume that $U$ is dense. We define $d(v)$ by induction on $v$ as follows. Let $v$ such that $d(v)$ is not yet defined, we pick a node $v i_1\cdots i_k \in U$ (there must exists one since $U$ is dense), and for all $j\leq k$ we define 
\[d(v i_1 \cdots i_{j-1})= i_j.\] The mapping is defined on every node and satisfies the requirement by definition. The other implication is straightforward (for all node $v$, there exists $v i_1\cdots i_k \in U$).
\end{proof}

Fix a tree $t$ and define an acceptance game for $\LLarge{\A}$ as follows. The game is played along a tree, \Eloise chooses the transitions of the automaton and \Abelard chooses the directions. Furthermore, at each node \Eloise proposes a direction that \Abelard may or may not follow, and possibly marks some of the sons of the current state. We keep track in \Eloise's vertices of informations about the choice of \Abelard in his previous move differentiating three possible situations:
\begin{enumerate}
\item[($\star$)] \Abelard has picked a son that \Eloise has marked,
\item[($\circ$)] \Abelard has not picked a marked son, but he has followed the direction that \Eloise has given,
\item[($\square$)] \Abelard has not picked a marked son and has not followed the direction given by \Eloise.
\end{enumerate}

Therefore \Eloise's vertices will be of the form $(q,u,\textit{symb})$ with $q$ a state, $u$ a node, and $\textit{symb}\in\{\star,\circ,\square\}$, and we define the colour of this vertex as the colour of $q$, and \Abelard's state will be of the form $(q,u,q_0,q_1,i,S)$ where $(q,t(u),q_0,q_1)$ is a transition of the automaton, $i\in\{0,1\}$ is the direction that \Eloise has proposed in the previous turn and $S\subseteq \{0,1\}$ describes which sons of $u$ she marked (see Figure~\ref{largeGame} for the local structure of the game).

The accepting condition for \Eloise is described as follows. She wins a play if and only if one of the following occurs.
\begin{itemize}
\item There are infinitely many $\star$-vertices and the smallest colour appearing infinitely often is even
\item \oschanged{Eventually there are no more  $\star$-vertices but there are infinitely often  $\square$-vertices, \ie \Abelard stop visiting marked nodes and avoids infinitely often the directions given by \Eloise.}
\end{itemize}

Call $\gameAccLarge{\A,t}$ this game.

Intuitively a strategy of \Eloise is a run of the automaton over the tree, along with a set $U$ of marked nodes and directions on each of the nodes, and \Abelard chooses a branch along the tree. If at some point \Abelard follows forever the directions given by \Eloise without going through a marked node, then  \Abelard wins. If \Abelard goes infinitely often through a marked node, then the smallest colour seen infinitely often is the one of the branch in the run of \Eloise, therefore \Eloise wins if this branch is accepting. These two remarks intuitively mean that if \Eloise has a winning strategy, then the set $U$ of marked nodes implied by this strategy must be a dense set and $\Brchs(U)$ must consist only of accepting branches of the run therefore the set of accepting branches of the run is large.

On the other hand, if there exists a run whose set of accepting branches is large, there exists a dense set of nodes $U$ such that all branches in $\Brchs(U)$ are accepting (Lemma~\ref{lemma:largeVSdense}), and directions on each nodes that leads to  nodes in $U$ (Lemma~\ref{directions}). If \Eloise plays according to them, she wins in the game. Indeed, if \Abelard follows infinitely often the nodes in $U$ then the branch is an accepting branch, therefore \Eloise wins the game.  His only option to avoid the nodes of $U$ is to infinitely often go in the opposite direction than the one given by \Eloise, in which case \Eloise also wins.

\begin{figure}
\begin{center}
\begin{tikzpicture}[scale=1,transform shape]
\tikzset{>=stealth}

\node[draw,circle] (star) at (-5,0) {$q,u,\star$};
\node[draw,circle] (circ) at (0,0) {$q,u,\circ$};
\node[draw,circle] (box) at (5,0) {$q,u,\square$};

\node[draw] (abelard) at (0,-3) {$q,u,q_0,q_1,i,S$};

\node[right,text width = 5cm] at (0,-2.15) {\footnotesize for any $(q,t(u),q_0,q_1)\in \Delta$, any $i\in \{0,1\}$ and any $S\subseteq \{0,1\}$};

\node[right] at (0,-3.7) {\footnotesize for any $j\in \{0,1\}$};

\node[above] at (-2.5,-4.5) {\scriptsize if $j\in S$};

\node[above] at (1.9,-4.5) {\scriptsize if $j\not\in S$};

\node[above] at (1.25,-5.5) {\scriptsize if $j=i$};

\node[above] at (3.75,-5.5) {\scriptsize if $j\neq i$};

\draw [rounded corners,->] (star) -- (-5,-1.5) -- (0,-1.5) -- (abelard);
\draw [rounded corners,->] (circ) --  (abelard);
\draw [rounded corners,->] (box) -- (5,-1.5) -- (0,-1.5) -- (abelard);

\node[draw,ellipse] (starj) at (-5,-6.5) {$q_j,u\cdot j,\star$};
\node[draw,ellipse] (circj) at (0,-6.5) {$q_j,u\cdot j,\circ$};
\node[draw,ellipse] (boxj) at (5,-6.5) {$q_j,u\cdot j,\square$};

\draw [rounded corners,->] (abelard) -- (0,-4.5) -- (-5,-4.5) -- (starj);
\draw [rounded corners,->] (abelard) -- (0,-4.5) --  (2.5,-4.5) -- (2.5,-5.5) -- (0,-5.5) -- (circj);
\draw [rounded corners,->] (abelard) -- (0,-4.5) --  (2.5,-4.5) -- (2.5,-5.5) -- (5,-5.5) -- (boxj);

\end{tikzpicture}
	\caption{Local structure of the arena of the acceptance game $\gameAccLarge{\A,t}$.\label{largeGame}}
\end{center}
\end{figure}

The next theorem states that $\gameAccLarge{\A,t}$ is an acceptance game for $\LLarge{\A}$.

\begin{theorem}\label{theo:LLarge_AcceptanceGame}
One has $t\in \LLarge{\A}$ if and only if \Eloise wins in $\gameAccLarge{\A,t}$ from $(\qini,\epsilon,\circ)$.
\end{theorem}

\begin{proof}
Assume that \Eloise has a winning strategy $\phi$ in $\gameAccLarge{\A,t}$ from $(\qini,\epsilon,\circ)$. With $\phi$ we associate a run $\run$ of $\A$ on $t$ as follows. We inductively associate with any node $u$ a partial play $\play_u$ where \Eloise respects $\phi$. For this we let $\play_\epsilon = (\qini,\epsilon,\circ)$. Now assume that we defined $\play_u$ for some node $u$ and let $(q,t(u),q_0,q_1)$ be the transition \Eloise plays from $\play_u$ when she respects $\phi$. 

For $j\in\{0,1\}$, one defines $\play_{u\cdot j}$ as the partial play obtained by extending $\play_u$ by \Eloise choosing transition $(q,t(u),q_0,q_1)$, followed by \Abelard choosing direction $j$. Note that for $j\in\{0,1\}$, $\play_{u\cdot j}$ ends in configuration $(q_j,uj,\textit{symb})$ for some $\textit{symb} \in \{\star,\circ,\square\}$.

The run $\run$ is defined by letting $\run(u)$ be the state $q$  in the last configuration $(q,u,\textit{symb})$ of $\play_u$. By construction, $\run$ is a valid run of $\A$ on $t$ and moreover with any branch $\pi$ in $\run$ one can associate a play $\play_\pi$ in $\gameAccLarge{\A,t}$ from $(\qini,\epsilon,\circ)$ where \Eloise respects $\phi$ (one simply considers the limit of the \emph{increasing} sequence of partial plays $\play_u$ where $u$ ranges over those nodes along branch $\pi$). By construction $\pi$ is accepting if and only if $\play_\pi$  fulfils the parity condition.

We define $s(u)$ as the symbol $\textit{symb}$ in the last configuration $(q,u,\textit{symb})$ of $\play_u$. Furthermore, we define a direction mapping $d$ and a set of nodes $U$ as follows: for all $u$, $d(u) = i$ with $\phi(\play_u) = (q,u,q_0,q_1,i,S)$; and for all $u$, $u\in U$ if and only if $s(u)=\star$. Notice that if $d(u)=i$ then $s(u\cdot i) = \circ$ or $s(u\cdot i) = \star$.

Given a branch $\pi=i_1 i_2 \cdots$ we define $s(\pi)$ as the infinite sequence of $s(\epsilon) s(i_1) s(i_1i_2) \cdots$, and $\col(\pi)$ as the smallest colours appearing infinitely often in $\run(\pi)$.  Note that since \Eloise wins the play $\play_\pi$, $\star$ appears infinitely often in $s(\pi)$ and $\col(\pi)$ is even, or $\star$ does not appear infinitely often in $s(\pi)$ but $\square$ does.

First let us show that $d$ points to $U$. Suppose by contradiction that this is not the case, \ie  there exists a branch $\pi = u i_1 i_2 \cdots $, with $i_j=d(u i_1 \cdots i_{j-1})$ for all $j\geq 1$, such that for all $k\geq 1$, $ui_1 \cdots i_k\not\in U$. Then for all $k\geq 1$, $s(ui_1 \cdots i_k)=\circ$, therefore $\play_\pi$ is loosing. This raises a contradiction since $\phi$ is a winning strategy.

Now, let us show that all branches in $\Brchs(U)$ are winning in $\run$. Let $\pi\in \Brchs(U)$. Then by definition, $\star$ appears infinitely often in $s(\pi)$. Then since $\play_\pi$ is winning we have that $\col(\pi)$ is even, then $\pi$ is an accepting branch in $\run$.

Conversely let $\run$ be a run whose set of accepting branches is large. From Lemma~\ref{directions} there exist  a direction mapping $d$ and a set of nodes $U$ such that $d$ points to $U$, and every branch $\pi \in \Brchs(U)$ is accepting in $\run$. Define the strategy $\phi$ of \Eloise as follows. For all partial play $\play$ ending in $(\run(u),u,\textit{symb})$  
\[\phi(\play) = (\run(u),u,\run(u0),\run(u1), d(u), \{ j \mid uj\in U\}),\]
and for all other plays, we do not give any restriction on $\phi(\play)$ (assuming that the automaton is complete, \Eloise can always play something). Let us show that $\phi$ is a winning strategy for \Eloise.

As for the other direction, we inductively associate with any node $u$ a partial play $\play_u$ where \Eloise respects $\phi$. For this we let $\play_\epsilon = (\qini,\epsilon,\circ)$. Now assume that we defined $\play_u$ for some node $u$ and let $(q,t(u),q_0,q_1)$ be the transition \Eloise plays from $\play_u$ when she respects $\phi$.  For $j\in\{0,1\}$, one defines $\play_{u\cdot j}$ as the partial play obtained by extending $\play_u$ by \Eloise choosing transition $(q,t(u),q_0,q_1)$, followed by \Abelard choosing direction $j$. Note that for $j\in\{0,1\}$, $\play_{u\cdot j}$ ends in configuration $(q_j,uj,\textit{symb})$ for some $\textit{symb} \in \{\star,\circ,\square\}$. 

Moreover with any branch $\pi$ in $\run$ one can associate a play $\play_\pi$ in $\gameAccLarge{\A,t}$ from $(\qini,\epsilon,\circ)$ where \Eloise respects $\phi$ (one simply considers the limit of the \emph{increasing} sequence of partial plays $\play_u$ where $u$ ranges over those nodes along branch $\pi$). By construction $\pi$ is accepting in $\run$ if and only if $\play_\pi$  fulfils the parity condition. Furthermore observe that any play that respects $\phi$ is equal to $\play_\pi$ for some branch $\pi$. Again, we define $s(u)$ as the symbol $\textit{symb}$ in the last configuration $(q,u,\textit{symb})$ of $\play_u$. Observe that if $s(u\cdot i) = \circ$ for some node $u$ and $i\in \{0,1\}$ then $i = d(u)$.

Let $\play_\pi$ be a play that respects $\phi$. Note that \Eloise wins the play $\play_\pi$ if and only if $\star$ appears infinitely often in $s(\pi)$ and $\col(\pi)$ is even, or $\star$ does not appear infinitely often in $s(\pi)$ but $\square$ does. First observe that $u\in U$ if and only if $s(u) = \star$. If $\star$ appears infinitely often in $s(\pi)$ then $\pi$ is in $\Brchs(U)$ therefore it is accepting, thus $\play_\pi$ is winning. If $\star$ does not appears infinitely often in $s(\pi)$ let $u$  and $i_1,i_2,\ldots$ be such that $\pi=u i_1 i_2 \cdots$ and for all $k$, $s(u i_1 \cdots i_k)\neq \star$. Assume by contradiction that $\square$ does not appears infinitely often in $s(\pi)$. Therefore there exists $\ell$ such that for all $k\geq \ell$, $s(u i_1 \cdots i_k) = \circ$, therefore $k = d(u i_1 \cdots i_{k-1})$. Thus $\pi$ is a branch where at some point $d$ is followed, but no node in $U$ is eventually reached, which means that $d$ does not points to $U$ hence, raises a contradiction.

Therefore $\play_\pi$ is a winning play, thus $\phi$ is a winning strategy.
\end{proof}

Thanks to Theorem~\ref{theo:LLarge_AcceptanceGame} we can now easily prove that language of the form $\LLarge{\A}$ are always $\omega$-regular.

\begin{theorem}\label{thm:LLarge:reg}
Let $\A=\langle A, Q , q_{ini},\Delta,\col \rangle$ be a parity tree automaton using $d$ colours. Then there exists a parity tree automaton $\A'=\langle A, Q' , q'_{ini},\Delta',\col' \rangle$ such that $\LLarge{\A}=L(\A')$. Moreover $|Q'|=\mathcal{O}(d|Q|)$ and $\A'$ uses $d+2$ colours.
\end{theorem}

\begin{proof}
The  game $\gameAccLarge{\A,t}$ can be transformed into a standard acceptance game for $\omega$-regular language (as defined in Section~\ref{ssection:TARTLAG}) by the following trick (this is the same as the one for $\gameRejCountB{\A,t}$). One adds to states an integer where one stores the smallest colour seen since the last $\star$-state was visited (this colour is easily updated); whenever a starred state is visited the colour is reseted to the colour of the state. Now $\square$-states are given an even colour $e$ that is greater or equal than all colour previously used (hence, it ensures that if finitely many $\star$-states but infinitely many $\square$-states are visited then \Eloise wins), $\circ$-states are given the odd colour $e+1$ (hence it ensures that if at some points only $\circ$-states are visited, \Eloise looses) and starred states are given the colour that was stored (hence, if infinitely many starred states are visited we retrieve the previous parity condition).
It should then be clear that the latter game is a classical acceptance game, showing that $\LLarge{\A}$ is $\omega$-regular. 

The construction of $\A'$ is immediate from the final game and the size is linear in $d|Q|$ due to the fact that one needs to compute the smallest colour visited between to starred states.
\end{proof}

\newcommand{\noopsort}[1]{} \newcommand{\singleletter}[1]{#1}
  \newcommand{\etal}{et al.}

\end{document}